\documentclass[12pt]{article}
\usepackage{graphicx}          
\usepackage{amsfonts,amsmath,amssymb}
\usepackage{mathrsfs}
\usepackage{appendix}
\newtheorem{theorem}{Theorem}

\newtheorem{example}[theorem]{Example}

\newtheorem{lemma}[theorem]{Lemma}

\newenvironment{proof}[1][Proof]{\noindent\textbf{#1.} }{\ }
\newenvironment{keywords}[1][Keywords]{\noindent\textbf{#1:} }{}

\renewcommand\appendix{\par
\setcounter{section}{0}%
\setcounter{subsection}{0}%
\setcounter{table}{0}
\setcounter{figure}{0}
\gdef\thetable{\Alph{table}}
\gdef\thefigure{\Alph{figure}}
\section*{Appendix}
\gdef\thesection{\Alph{section}}
\setcounter{section}{1}}

\begin{document}

\title{Structures and Transformations for Model Reduction of Linear Quantum Stochastic Systems\footnote{Research supported by the Australian Research Council}}
\author{Hendra~I.~Nurdin 
\thanks{Hendra I. Nurdin is with the School of Electrical Engineering and Telecommunications, The University of New South Wales (UNSW), 
Sydney NSW 2052, Australia. Email: h.nurdin@unsw.edu.au}}
\date{}

\maketitle \thispagestyle{empty}

\begin{abstract}
The purpose of this paper is to develop a model reduction theory for linear quantum stochastic systems that are commonly encountered in quantum optics and related fields, modeling devices such as optical cavities and optical parametric amplifiers, as well as quantum networks composed of such devices. Results are derived on subsystem truncation of such systems and it is shown that this truncation preserves the physical realizability property of linear quantum stochastic systems. It is also shown that the property of complete passivity of linear quantum stochastic systems is preserved under subsystem truncation. A necessary and sufficient condition for the existence of a balanced realization of a linear quantum stochastic system under sympletic transformations is derived. Such a condition turns out to be very restrictive and will not be satisfied by generic linear quantum stochastic systems, thus necessary and sufficient conditions for relaxed notions of simultaneous diagonalization of the controllability and observability Gramians of  linear quantum stochastic systems under symplectic transformations are also obtained. The notion of a quasi-balanced realization is introduced and it is shown that all asymptotically stable completely passive linear quantum stochastic systems have a quasi-balanced realization. Moreover, an explicit bound for the subsystem truncation error on a quasi-balanceable linear quantum stochastic system  is provided. The results are applied in an example of model reduction in the context of low-pass optical filtering of coherent light using a network of optical cavities.
\end{abstract}


\section{Introduction}
\label{sec:intro}

The class of linear quantum stochastic systems \cite{BE08,JNP06,NJD08,GJN10} represents multiple distinct open quantum harmonic oscillators that are coupled linearly to one another and also to external Gaussian fields, e.g., coherent laser beams, and whose dynamics can be conveniently and completely summarized in the Heisenberg picture of quantum mechanics in terms of a quartet of matrices $A,B,C,D$, analogous to those used in modern control theory for linear systems. As such, they can  be viewed as a quantum analogue of classical linear stochastic systems and are encountered in practice, for instance, as models for optical parametric amplifiers \cite[Chapters 7 and 10]{GZ04}. However, due to the constraints imposed by quantum mechanics, the matrices $A,B,C,D$ in a linear quantum stochastic system cannot be arbitrary, a restriction not encountered in the classical setting. In fact, as derived in \cite{JNP06} for the case where $D$ is of the form $D=[\begin{array}{cc} I & 0 \end{array}]$, with $I$ denoting an identity matrix, it is required that $A$ and $B$ satisfy a certain non-linear equality constraint, and $B$ and $C$ satisfy a linear equality constraint. These constraints on the $A,B,C,D$ matrices are referred to as {\em physical realizability} constraints \cite{JNP06}. Due to the analogy with classical linear stochastic systems, linear quantum stochastic systems provide a particularly tractable class of quantum systems with which to discover and develop fundamental ideas and principles of quantum control, just as classical linear systems played a fundamental role in  the early development of systems and control theory. 

In control problems involving linear quantums stochastic systems such as $H^{\infty}$ control \cite{JNP06} and LQG control \cite{NJP07b}, the important feature of the controller is its transfer function rather than the systems matrices $(A,B,C,D)$. The controller may have many degrees of freedom, which may make it challenging to realize. Therefore it is of interest to have a method to  construct an approximate controller with a smaller number of degrees of freedom whose transfer function approximates that of the full controller. In systems and control theory, this procedure is known as model reduction and is an important part of  a controller design process, see, e.g. \cite{ZDG95}. 

Model reduction methods for linear quantum stochastic systems have been limited to singular perturbation techniques \cite{BvHS07,GNW10,Pet10} and an eigenvalue truncation technique that is restricted to a certain sub-class of completely passive linear quantum stochastic systems \cite{Pet12}. These methods cannot be applied to general linear quantum stochastic systems and the current paper contributes towards filling this important gap  by developing new results on subsystem truncation for general linear quantum stochastic systems. Moreover, the paper studies the feasibility of performing model reduction by balanced truncation for linear quantum stochastic systems and derives a necessary and sufficient condition  under which it can be carried out. It is shown that  balanced truncation is {\em not} possible for generic linear quantum stochastic systems. Therefore, this paper also considers other realizations in which the system controllability and observability Gramians are simultaneously diagonal, and introduces one such realization which is referred to as a quasi-balanced realization. The results are illustrated in an example that demonstrates an instance where quasi-balanced truncation can be applied.   

\section{Preliminaries}
\label{sec:preliminaries}

\subsection{Notation}
We will use the following notation: $\imath=\sqrt{-1}$, $^*$ denotes the adjoint of a linear operator
as well as the conjugate of a complex number. If $A=[a_{jk}]$ then $A^{\#}=[a_{jk}^*]$, and $A^{\dag}=(A^{\#})^{\top}$, where $(\cdot)^{\top}$ denotes matrix transposition.  $\Re\{A\}=(A+A^{\#})/2$ and $\Im\{A\}=\frac{1}{2\imath}(A-A^{\#})$.
We denote the identity matrix by $I$ whenever its size can be
inferred from context and use $I_{n}$ to denote an $n \times n$
identity matrix. Similarly, $0_{m \times n}$ denotes  a $m \times n$ matrix with zero
entries but drop the subscript when its dimension  can be determined from context. We use
${\rm diag}(M_1,M_2,\ldots,M_n)$  to denote a block diagonal matrix with
square matrices $M_1,M_2,\ldots,M_n$ on its diagonal, and ${\rm diag}_{n}(M)$
denotes a block diagonal matrix with the
square matrix $M$ appearing on its diagonal blocks $n$ times. Also, we
will let $\mathbb{J}=\left[\begin{array}{rr}0 & 1\\-1&0\end{array}\right]$ and $\mathbb{J}_n=I_{n} \otimes \mathbb{J}={\rm diag}_n(\mathbb{J})$.

\subsection{The class of linear quantum stochastic systems}
\label{sec:linear-summary} 

Let $x=(q_1,p_1,q_2,p_2,\ldots,q_n,p_n)^T$ denote a vector of  the canonical position and
momentum operators of a {\em many degrees of freedom quantum
harmonic oscillator} satisfying the canonical commutation
relations (CCR)  $xx^T-(xx^T)^T=2\imath \mathbb{J}_n$. A {\em linear quantum stochastic
system} \cite{JNP06,NJP07b,NJD08} $G$ is a quantum system defined by three {\em parameters}:
(i) A quadratic Hamiltonian $H=\frac{1}{2} x^T R x$ with $R=R^T \in
\mathbb{R}^{n \times n}$, (ii) a coupling operator $L=Kx$, where $K$ is
an $m \times 2n$ complex matrix, and (iii) a unitary $m \times m$
scattering matrix $S$. For shorthand, we write $G=(S,L,H)$ or $G=(S,Kx,\frac{1}{2} x^TRx)$. The
time evolution $x(t)$ of $x$ in the Heisenberg picture ($ t
\geq 0$) is given by the quantum stochastic differential equation
(QSDE) (see \cite{BE08,JNP06,NJD08}):
\begin{align}
dx(t) &=  A_0x(t)dt+ B_0\left[\small \begin{array}{c} d\mathcal{A}(t)
\\ d\mathcal{A}(t)^{\#} \end{array} \normalsize \right];  
x(0)=x, \notag\\
dY(t) &=  C_0 x(t)dt+  D_0 d\mathcal{A}(t), \label{eq:qsde-out}
\end{align}
with $A_0=2\mathbb{J}_n(R+\Im\{K^{\dag}K\})$, $B_0=2\imath \mathbb{J}_n [\begin{array}{cc}
-K^{\dag}S & K^TS^{\#}\end{array}]$,
$C_0=K$, and $D_0=S$. Here $Y(t)=(Y_1(t),\ldots,Y_m(t))^{\top}$ is a vector of
continuous-mode bosonic {\em output fields} that results from the interaction of the quantum
harmonic oscillators and the incoming continuous-mode bosonic quantum fields in the $m$-dimensional vector
$\mathcal{A}(t)$. Note that the dynamics of $x(t)$ is linear, and
$Y(t)$ depends linearly on $x(s)$, $0 \leq s \leq t$. We refer to $n$ as the {\em
  degrees of freedom} of the system or, more simply, the {\em degree} of the system.

Following \cite{JNP06}, it will be  convenient to write the dynamics in quadrature form as
\begin{align}
dx(t)&=Ax(t)dt+Bdw(t);\, x(0)=x. \nonumber\\
dy(t)&= C x(t)dt+ D dw(t), \label{eq:qsde-out-quad}
\end{align}
with
\begin{eqnarray*}
w(t)&=&2(\Re\{\mathcal{A}_1(t)\},\Im\{\mathcal{A}_1(t)\},\Re\{\mathcal{A}_2(t)\},\Im\{\mathcal{A}_2(t)\},\ldots,\Re\{\mathcal{A}_m(t)\},\Im\{\mathcal{A}_m(t)\})^{\top}; \\
y(t)&=&2(\Re\{Y_1(t)\},\Im\{Y_1(t)\},\Re\{Y_2(t)\},\Im\{Y_2(t)\}, \ldots,\Re\{Y_m(t)\},\Im\{Y_m(t)\})^{\top}.
\end{eqnarray*}
The real matrices $A,B,C,D$ are in a one-to-one correspondence
with $A_0,B_0,C_0,D_0$. Also, $w(t)$ is taken to be in a vacuum state where it
satisfies the It\^{o} relationship $dw(t)dw(t)^{\top} = (I+\imath \mathbb{J}_m)dt$; see \cite{JNP06}. Note that in this form it follows that $D$ is a real unitary symplectic matrix. That is, it is both unitary (i.e., $DD^{\top}=D^{\top} D=I$) and symplectic (a real $m \times m$ matrix is symplectic if $D \mathbb{J}_m D^{\top} =\mathbb{J}_m$). However, in the most general case, $D$ can be generalized to a symplectic matrix that represents a quantum network that includes ideal squeezing devices acting on the incoming field $w(t)$ before interacting with the system \cite{GJN10,NJD08}.  The matrices $A$, $B$, $C$, $D$ of a linear quantum stochastic system cannot be arbitrary and are not independent of one another. In fact, for the system to be physically realizable \cite{JNP06,NJP07b,NJD08}, meaning it represents a meaningful physical system, they must satisfy the constraints (see \cite{WNZJ12,JNP06,NJP07b,NJD08,GJN10})
\begin{eqnarray}
&&A\mathbb{J}_n + \mathbb{J}_n A^{\top} + B\mathbb{J}_mB^{\top}=0, \label{eq:pr-1}\\
&& \mathbb{J}_n  C^{\top} +B\mathbb{J}_mD^{\top}=0, \label{eq:pr-2}\\
&&D \mathbb{J}_m D^{\top} =  \mathbb{J}_{m}. \label{eq:pr-3}
\end{eqnarray}

The above are the physical realizability constraints for systems for which the (even) dimension of the output $y(t)$ is the same as that of the input $w(t)$, i.e., $n_y=2m$. However, for the purposes of the model reduction theory to be developed in this paper, it is pertinent to consider the case where $y(t)$ has an even dimension possibly less than $w(t)$. The reason for this and the physical realizability constraints for systems with less outputs and inputs are given in the next section. 

Following \cite{GJ07}, we denote a  linear quantum stochastic system having an equal number of inputs and outputs, and  Hamiltonian $H$,  coupling vector $L$, and scattering matrix $S$, simply as $G=(S, L,H)$ or $G=(S,Kx,\frac{1}{2} x^{\top}Rx)$. We also recall  the {\em concatenation product} $\boxplus$ and {\em series product} $\triangleleft$  for open Markov quantum systems \cite{GJ07} defined by $G_1 \boxplus G_2=({\rm diag}(S_1,S_2),(L_1^{\top},L_2^{\top})^{\top},H_1+H_2)$, and $G_2 \triangleleft G_1=(S_2S_1, L_2+S_2   L_1,H_1+H_2+\Im\{L_2^{\dag}S_2 L_1\})$. Since both products are associative, the products $G_1
\boxplus G_2 \boxplus \ldots \boxplus G_n$ and $G_n \triangleleft G_{n-1} \triangleleft \ldots \triangleleft G_1$ are unambiguously defined. 


\subsection{Linear quantum stochastic systems with less outputs than inputs}
\label{sec:pr-less-out}
In general one may not be interested in all outputs of the system but only in a subset of them, see, e.g., \cite{JNP06}. That is, one is often only interested in certain pairs of the output field quadratures in $y(t)$. Thus, in the most general scenario, $y(t)$ can have an even dimension $n_y <2m$ and $D$ is a $n_y \times 2m$ matrix satisfying $D \mathbb{J}_m  D^{\top}=\mathbb{J}_{n_y/2}$. Thus, more generally we can consider outputs $y(t)$ of form
\begin{equation}
y(t) = C x(t) + D w(t), \label{eq:y-e}
\end{equation}
with $C \in \mathbb{R}^{n_y \times  2n}$, $D \in \mathbb{R}^{n_y \times 2m}$ with $n_y$ even and $n_y < 2m$. In this case, generalizing the notion developed in \cite{JNP06}, we say that a linear quantum stochastic system with output (\ref{eq:y-e}) is physically realizable if and only if there exists matrices $C' \in \mathbb{R}^{(2m-n_y) \times 2n}$ and $D' \in \mathbb{R}^{(2m-n_y) \times 2m}$ such that the system 
\begin{align}
dx(t)&=Ax(t)dt+Bdw(t);\, x(0)=x. \nonumber\\
dy'(t)&= \left[\begin{array}{c} C \\ C' \end{array}\right]x(t)dt+ \left[\begin{array}{c} D \\ D' \end{array}\right]dw(t), \label{eq:qsde-out-quad-e}
\end{align}
is a physically realizable linear quantum stochastic system with the same number of inputs and outputs. That is, the matrices $A$, $B$, $[\begin{array}{cc} C^{\top} & (C')^{\top} \end{array}]^{\top}$, and $[\begin{array}{cc} D^{\top} & (D')^{\top} \end{array}]^{\top}$ satisfy the constraints (\ref{eq:pr-1})-(\ref{eq:pr-3}) when $C$ and $D$ in (\ref{eq:pr-2}) and (\ref{eq:pr-3}) are replaced by $[\begin{array}{cc} C^{\top} & (C')^{\top} \end{array}]^{\top}$ and $[\begin{array}{cc} D^{\top} & (D')^{\top} \end{array}]^{\top}$, respectively. A necessary and sufficient condition for physical realizability of general linear quantum stochastic systems is the following \cite{WNZJ12}:
\begin{theorem}
A linear quantum stochastic system with less outputs than inputs is physically realizable if and only if
\begin{eqnarray}
&&A\mathbb{J}_n + \mathbb{J}_n A^{\top} + B\mathbb{J}_mB^{\top}=0, \label{eq:pr-1e}\\
&& \mathbb{J}_n  C^{\top} +B\mathbb{J}_mD^{\top}=0, \label{eq:pr-2e}\\
&&D \mathbb{J}_m D^{\top} =  \mathbb{J}_{n_y/2}. \label{eq:pr-3e}
\end{eqnarray}
\end{theorem}
A proof of this theorem had to be omitted in \cite{WNZJ12} due to page limitations, so a short independent proof is provided below.

\begin{proof}
The necessity of (\ref{eq:pr-1e})-(\ref{eq:pr-3e}) follows immediately from  the definition  of a physically realizable system with less outputs than inputs (as given previously) and from the physical realizability contraints for systems with the same number of input and outputs. As for the sufficiency, first note that for $D$ satisfying (\ref{eq:pr-3e}), it follows from an analogous  construction to that given in the proof of \cite[Lemma 6]{Nurd11} that a matrix $D' \in \mathbb{R}^{(2m-n_y)\times 2m}$ can be constructed such that the the matrix $\tilde D = [\begin{array}{cc} D^{\top} & (D')^{\top} \end{array}]^{\top}$ is symplectic. Now, define $C' = -D' \mathbb{J}_m B^{\top}$ and $\tilde C=[\begin{array}{cc} C^{\top} & (C')^{\top} \end{array}]^{\top}$. Consider now a system $\tilde G$ with an equal number of inputs and outputs, and system matrices $(A,B,\tilde C, \tilde D)$. From the physical realizability conditions (\ref{eq:pr-1e})-(\ref{eq:pr-3e}) and the definition of $C'$ and $\tilde C$, it follows that $\tilde G$ satisfies (\ref{eq:pr-1e})-(\ref{eq:pr-3e}) and is therefore physically realizable with the same number of inputs and outputs. It now follows from  definition that the original system with output $y'(t)$ of smaller dimension that $w(t)$ is physically realizable. This completes the proof. $\hfill \Box$
\end{proof}

\section{Model reduction of linear quantum stochastic systems by subsystem truncation}
\label{sec:subsys-truncation}
\subsection{Preservation of quantum structural constraints in subsystem truncation}
\label{sec:pr-preserve-truncation}
In this section we show that physically realizable linear quantum stochastic systems possess  the convenient property that any subsystem defined by a collection of arbitrary pairs $(q_j,p_j)$ in $x$ and obtained via a simple truncation procedure inherit the physical realizability property. 

Let $\pi$ be any permutation map on $\{1,2,\ldots,n\}$, i.e., a bijective map of  $\{1,2,\ldots,n\}$ to itself. Let $x_{\pi}=(q_{\pi(1)},
p_{\pi(1)},q_{\pi(2)},p_{\pi(2)},\ldots,q_{\pi(n)},p_{\pi(n)})^{\top}$, and $P$ be the permutation matrix representing this permutation of the elements of $x$, i.e., $Px=x_{\pi}$. Then the permuted system $G_{\pi}$ will have system matrices 
$(A_{\pi},B_{\pi},C_{\pi},D_{\pi})$ with $A_{\pi}=PA P^{\top}$, $B_{\pi}=PB$, $C_{\pi}=CP^{\top}$, $D_{\pi}=D$. Since $G_{\pi}$ involves a mere rearrangement of  the degrees of freedom  $x$ of  $G$, it represents the same  physically realizable system as $G$, up to a reordering of the components of $x$. Thus the system matrices $(A_{\pi},B_{\pi},C_{\pi},D_{\pi})$ of $G_{\pi}$ trivially satisfy the physically realizability constraints (\ref{eq:pr-1})-(\ref{eq:pr-3}).  Partition $x_{\pi}$ as $x_{\pi}=(x_{\pi,1}^{\top},x_{\pi,2}^{\top})^{\top}$ where $x_{\pi,1}=(q_{\pi(1)},p_{\pi(1)},\ldots,q_{\pi(r)},p_{\pi(r)})^{\top}$  and $x_{\pi,2}=(q_{\pi(r+1)},p_{\pi(r+1)},\ldots,q_{\pi(n)},p_{\pi(n)})^{\top}$, with  $r<n$. Partition the matrices $A_{\pi}$, $B_{\pi}$, and $C_{\pi}$ compatibly with the partitioning of $x_{\pi}$ into $x_{\pi,1},x_{\pi,2}$. That is,
\begin{eqnarray}
A_{\pi}&=&\left[\begin{array}{cc} A_{\pi,11} & A_{\pi,12} \\ A_{\pi,21} & A_{\pi,22} \end{array} \right],\; B_{\pi} =\left[\begin{array}{c} B_{\pi,1} \\ B_{\pi,2} \end{array} \right], \label{eq:mat-part-1}\\
C_{\pi}&=&\left[\begin{array}{cc} C_{\pi,1} & C_{\pi,2} \end{array}\right]. \label{eq:mat-part-2}
\end{eqnarray}
From the fact that $A_{\pi}$, $B_{\pi}$, $C_{\pi}$, and $D_{\pi}$ satisfy the physical realizability constraints (\ref{eq:pr-1e})-(\ref{eq:pr-3e}) we immediately obtain for $j=1,2$:
\begin{eqnarray}
&&A_{\pi,jj}\Theta_{j}+\Theta_{j}A_{\pi,jj}^{\top} + B_{\pi,j}\mathbb{J}_m B_{\pi,j}^{\top} =0, \label{eq:part-AB}\\
&& \Theta_j C_{\pi,j}^{\top} +B_{\pi,j}  \mathbb{J}_m D_{\pi}^{\top}=0,\label{eq:part-BC} \\
&& D_{\pi} \mathbb{J}_m D_{\pi}^{\top} = \mathbb{J}_{n_y/2}, \label{eq:part-D}
\end{eqnarray}
where $\Theta_1 =\mathbb{J}_r$ and $\Theta_2=\mathbb{J}_{n-r}$. Therefore, the subsystems $G_{\pi,j}=(A_{\pi,jj},B_{\pi,j},C_{\pi,j},D_{\pi})$ with $x_{\pi,j}$ as canonical internal variable are physically realizable systems in their own right for $j=1,2$. Thus, we can state the following theorem.
\begin{theorem}
\label{thm:pr-sub-systems} For any given permutation map $\pi$ of the indices $\{1,2,\ldots,n\}$ and any partitioning of $x_{\pi}=(q_{\pi(1)},
p_{\pi(1)},q_{\pi(2)},p_{\pi(2)},\ldots,q_{\pi(n)},p_{\pi(n)})^{\top}$ as $x_{\pi}=(x_{\pi,1}^{\top},x_{\pi,2}^{\top})^{\top}$, with $x_{\pi,1}=(q_{\pi(1)},p_{\pi(1)},\ldots,q_{\pi(r)},p_{\pi(r)})^{\top}$  and $x_{\pi,2}=(q_{\pi(r+1)},p_{\pi(r+1)},\ldots,q_{\pi(r)},p_{\pi(r)})^{\top}$, with $r<n$, the subsystems $G_{\pi,j}=(A_{\pi,jj},B_{\pi,j},C_{\pi,j},D_{\pi})$ with canonical  position and momentum operators in $x_{\pi,j}$ are physically realizable for $j=1,2$.
\end{theorem}

From a model reduction perspective, the theorem says that if one truncates a subsystem $x_{\pi,j}$ according to any partitioning of $x_{\pi}$ in which each partition  $x_{\pi,1}$ and $x_{\pi,2}$  contain distinct pairs of conjugate position and momentum quadratures, then the remaining subsystem after the truncation (i.e., $x_{\pi,1}$ if  $x_{\pi,2}$ is truncated, and   $x_{\pi,2}$ if  $x_{\pi,1}$ is truncated) is automatically guaranteed to be a physically realizable linear quantum stochastic system. This is rather fortunate as the physical realizability constraints are quite formidable to deal with (see, e.g., \cite{NJP07b} in the context of coherent-feedback LQG controller design) and at a glance one would initially expect that physically realizable reduced models would not be easily obtained.

\begin{figure}[tbph]
\centering
\includegraphics[scale=0.5]{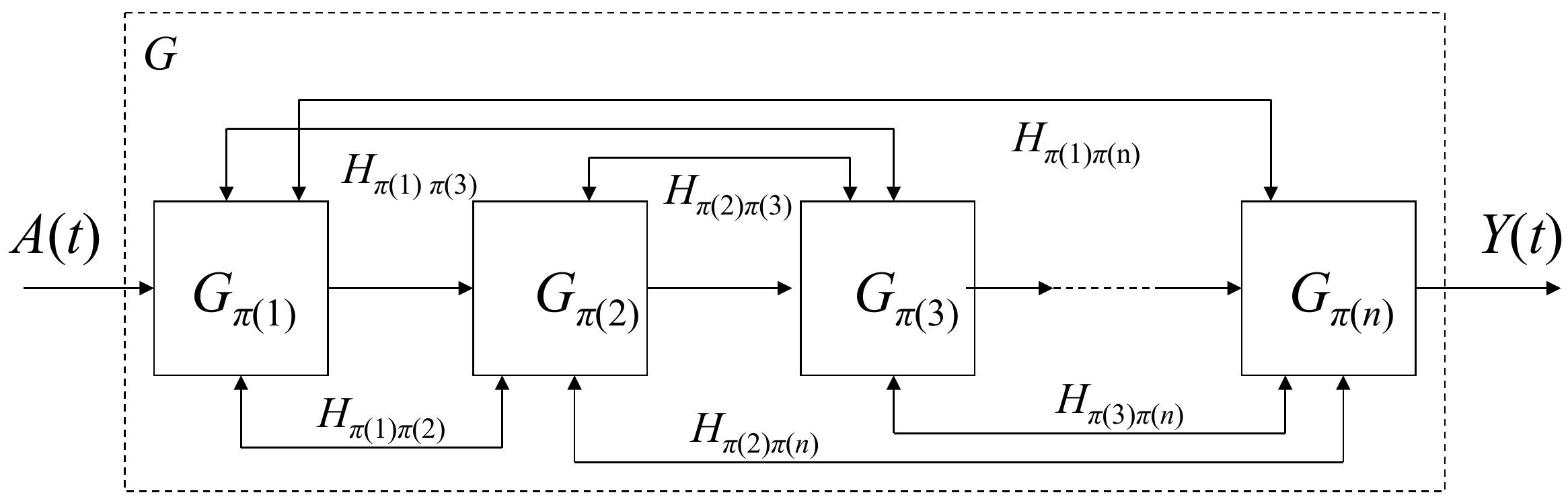}
\caption{Cascade realization of $G_{\pi}$ with direct interaction Hamiltonians $H^d_{\pi(j)\pi(k)}$ between sub-systems $G_{\pi(j)}$ and $G_{\pi(k)}$ for $j,k=1,2,\ldots,n$, following \cite{NJD08}. Illustration is for $n>3$.}
\label{fig:sysdecom}
\end{figure}

An equivalent proof of  the theorem is via the main network synthesis result of \cite{NJD08} -- this viewpoint of Theorem \ref{thm:pr-sub-systems} will be especially useful in the next section. It is shown in \cite[Theorem 5.1]{NJD08} that any (physically realizable) linear quantum stochastic system of degree $n$ such as  $G_{\pi}$  can be decomposed into a cascade or series connection of $n$ one degree of freedom linear quantum stochastic systems $G_{\pi(j)}$, $j=1,2,\ldots,n$ together with a direct quadratic coupling Hamiltonian between (at most) every pair of the $G_{\pi(j)}$'s, see Fig.~\ref{fig:sysdecom}. Here $G_{\pi(j)}$ is a one degree of freedom linear quantum stochastic system with $x_{\pi(j)}=(q_{\pi(j)},p_{\pi(j)})^{\top}$ as its canonical position and momentum operators. In the figure, $H^d_{\pi(j)\pi(k)}$ indicates the quadratic coupling Hamiltonian between $G_{\pi(j)}$ and $G_{\pi(k)}$. It shows that if we
\begin{enumerate}
\item remove  the  $n-r$ one degree of freedom subsystems $G_{\pi(r+1)}$, $G_{\pi(r+2)}$, $\ldots$, $G_{\pi(n)}$ from this cascade connection, 

\item remove all  Hamiltonian coupling terms associated with each of the subsystems that have been removed,

\begin{figure}[tbph]
\centering
\includegraphics[scale=0.5]{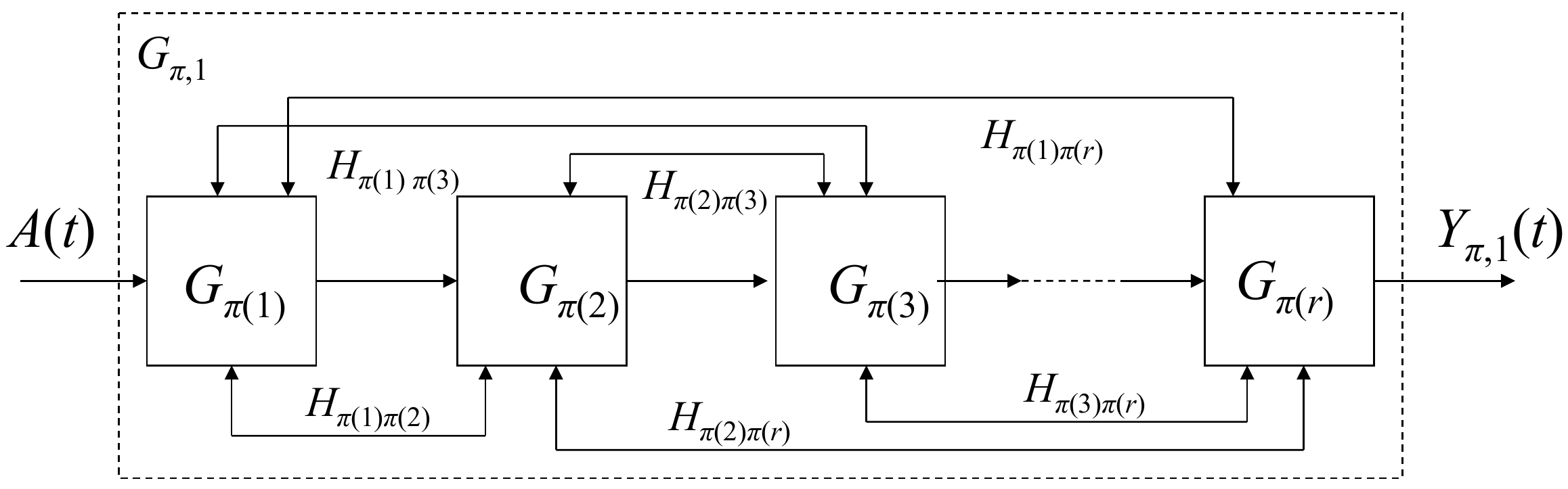}
\caption{Cascade realization of $G_{\pi,1}$ with direct interaction Hamiltonians $H^d_{\pi(j)\pi(k)}$ between sub-systems $G_{\pi(j)}$ and $G_{\pi(k)}$ for $j,k=1,2,\ldots,r$, following \cite{NJD08}. Illustration is for $r>3$.}
\label{fig:truncsubsysdecom}
\end{figure}

\item reconnect the remaining $r$ subsystems in a cascade connection in the  same order in which they appeared in the original cascade connection, and keeping the coupling Hamiltonians between each pair of remaining one degree of freedom sub-systems, as shown in Fig.~ \ref{fig:truncsubsysdecom},  
\end{enumerate}
we recover a physically realizable linear quantum stochastic system of degree $r$ as constructed in Theorem \ref{thm:pr-sub-systems}.

The theorem  may also be applied to the case where we allow certain transformations of $x(t$), namely symplectic transformations that preserve the canonical commutation relations (recall from Section \ref{sec:linear-summary} that a $2n \times 2n$ matrix is symplectic if $V \mathbb{J}_n V^{\top}=\mathbb{J}_n$. If $V$ is symplectic then so is $V^{\top}$ and $V^{-1}$). That is, we can transform internal variables from $x(t)$ to $z(t)=Vx(t)$, with $V$ symplectic so that $z(t)z(t)^{\top}-(z(t)z(t)^{\top})^{\top} = V(x(t)x(t)^{\top}-(x(t)x(t)^{\top})^{\top})V^{\top} =2\imath V \mathbb{J}_n V^{\top} =2 \imath \mathbb{J}_n $.  That is, $z(t)$ satisfies the same canonical commutation relations as $x(t)$. The dynamics of a system with $z(t)$ as the internal variable is then given by 
\begin{eqnarray*}
dz(t)&=&VAV^{-1} z(t) dt + VB dw(t),\; z(0)=z_0=Vx_0, \\
dy(t)&=& CV^{-1} z(t) + D dw(t),
\end{eqnarray*}
and again represents a physically realizable system. However,  strictly speaking, a linear quantum stochastic system with $x(t)$ as the internal variable and another system with $z(t)$ as the internal variable represent physically inequivalent quantum mechanical systems, although they have the same transfer function given by $C(sI-A)^{-1}B+D$. This physical subtlety, not encountered in the classical setting when similarity transformations are applied, has been discussed in some detail in \cite{Nurd10b}. In particular, the two systems do not have the same $S,L,H$ parameters.

If we are only interested in the steady-state input-output  evolution of $y(t)$ in relation to the driving noise $w(t)$ as $t \rightarrow \infty$ (assuming that the matrix $A$ is Hurwitz) then how the canonical position and momentum operators in $x(t)$  evolve is inconsequential. Thus, in this case we can allow a similarity transformation of the matrices $(A,B,C,D)$ to $(VAV^{-1}, VB, CV^{-1},D)$ with a symplectic $V$; see \cite{Nurd10b}. 
The advantage of such a transformation when we are mainly interested in steady-state input-output phenomena is that the transformed system matrices may be of  a more convenient form for analysis and computation, possibly allowing simplified formulas. Since $G'=(VAV^{-1}, VB, CV^{-1},D)$ is also a physically realizable system we can again apply Theorem \ref{thm:pr-sub-systems} to truncate certain sub-systems of $G'$.

\subsection{Application to completely passive linear quantum stochastic systems}
\label{sec:cp-systems}
We now specialize  to a class that will be referred to as  {\em completely passive} linear quantum stochastic systems \cite{GJN10,Pet10,Nurd10a,Nurd10b}. 
Following \cite{Nurd10a}, a physically realizable linear quantum stochastic system (\ref{eq:qsde-out-quad}) with an equal number of inputs and outputs is said to be completely passive if (i) $H$ can be written as $H = \frac{1}{2}a^{\dag} \tilde R a +c$, (ii) $L$ can be written as $L=\tilde Ka$ with $a=\frac{1}{2}(q_1+ \imath p_1,q_2+\imath p_2,\ldots,q_n + \imath p_n)^{\top}$ for some complex Hermitian matrix $\tilde R \in \mathbb{C}^{n \times n}$, a real constant $c$, and some $\tilde K \in \mathbb{C}^{m \times n}$, and  (iii) $D$ is unitary symplectic. On the other hand, if the system is of the form (\ref{eq:qsde-out-quad-e}) with less outputs than inputs, besides the same requirements (i) and (ii) of $H$ and $L$, for complete passivity we require that there exists a real matrix $E \in \mathbb{R}^{(2m-n_y)\times 2m}$ such that the matrix $\tilde D = [\begin{array}{cc} D^{\top} & E^{\top}\end{array}]^{\top}$ is unitary symplectic. Note that the latter systems are merely completely passive systems with an equal number of inputs and outputs  with certain pairs of output quadratures being ignored.  

It has been shown in \cite{Nurd10a}  that any completely passive system can be synthesized using purely passive devices, that is, devices that do not need an external source of quanta/energy. In quantum optics this means that they can be constructed using only optical cavities, beam splitters, and phase shifters.  We now show that the property of completely passivity  is also preserved under  subsystem truncation. The proof  is similar to that of \cite[Theorem 7]{Nurd10b}. 
\begin{lemma}
\label{lem:cp-preservation} If $G$ is completely passive then so is the truncated system $G_{\pi,1}$ for any permutation $\pi$.
\end{lemma}
\begin{proof}
Since $G$ is completely passive so is $G_{\pi}$ for any permutation $\pi$ because they represent the same physical system up to a permutation of the position and momentum operators. It suffices to consider completely passive systems with the same number of inputs and outputs, as any completely passive system with less outputs than inputs can be obtained from the former simply by disregarding pairs of output quadratures that are of no interest. To this end, assume that the system has an equal number of inputs and outputs and $S=I$ (i.e., $D=I$). Let $\tilde K=[\begin{array}{cccc} \tilde K_{1} & \tilde K_2 & \ldots & \tilde K_n\end{array}]$ and $\tilde R=[\tilde R_{jk}]_{j,k=1,2,\ldots,n}$, where $\tilde K_j \in \mathbb{C}^{m \times 1}$, and $\tilde R_{jk}$ are complex numbers with $\tilde R_{kj}=\tilde R_{jk}^*$.
Let $\tilde K_{\pi}=[\begin{array}{cccc} \tilde K_{\pi(1)} & \tilde K_{\pi(2)} & \ldots & \tilde K_{\pi(n)} \end{array}]$, $a_{\pi(j)}=\frac{1}{2}(q_{\pi(j)}+\imath p_{\pi(j)})$,  $G_{\pi(j)}=(I,\tilde K_{\pi(j)} a_{\pi(j)}, \frac{1}{2}\tilde R_{jj}a_{\pi(j)}^* a_{\pi(j)}+\tilde R_{jj}/4)$, and $H^d_{\pi(k)\pi(j)}=\tilde R_{kj}a_{\pi(k)}^*a_{\pi(j)}+\tilde R_{kj}^*a_{\pi(k)}a_{\pi(j)}^*+\frac{\imath}{2}(\tilde K_{\pi(k)}^{\dag}\tilde K_{\pi(j)}a_{\pi(k)}^*a_{\pi(j)}-\tilde K_{\pi(j)}^{\dag}\tilde K_{\pi(k)}a_{\pi(k)}a_{\pi(j)}^*)$  for all $k>j$ and $j=1,2,\ldots,n$. Then by \cite[Theorem 5.1]{NJD08}, we have that $G_{\pi}=(G_{\pi(n)} \triangleleft  \cdots \triangleleft G_{\pi(2)} \triangleleft G_{\pi(1)}) \boxplus (0,0,\sum_{j=1}^{n-1}\sum_{k=j+1}^n H^d_{\pi(k)\pi(j)})$ (recall the definition of the series product $\triangleleft$ and the concatenation product $\boxplus$ from Section \ref{sec:linear-summary}); see Fig.~\ref{fig:sysdecom}. Note that by construction all the $G_{\pi(j)}$'s are completely passive. Following the discussion in Sec. \ref{sec:pr-preserve-truncation}, we can write $G_{\pi,1}=(G_{\pi(r)} \triangleleft  \cdots \triangleleft G_{\pi(2)} \triangleleft G_{\pi(1)}) \boxplus (0,0,\sum_{j=1}^{r-1}\sum_{k=j+1}^r H^d_{\pi(k)\pi(j)})$. Since $G_{\pi(r)} \triangleleft  \cdots \triangleleft G_{\pi(2)} \triangleleft G_{\pi(1)}$ is by inspection completely passive, it is now apparent that $G_{\pi,1}$ is completely passive. Evidently this holds true for any permutation map $\pi$ since the choice of $\pi$ was arbitrary to begin with.

If $S$ is unitary  but not equal to the identity matrix (this means that $D$ is a unitary symplectic matrix different from the identity matrix), then one simply inserts a static passive network that implements $S$ between the input fields and $G_{\pi(1)}$; see \cite[Section 3]{NJD08}. The same argument as above then goes through.    $\hfill \Box$
\end{proof}

A truncation method has been proposed for a class of completely passive linear quantum stochastic systems in \cite{Pet12} based on an algorithm developed in \cite{Pet11}. This algorithm is not guaranteed to be applicable to all completely passive linear quantum stochastic systems but to a ``generic'' subclass of it. Theorem \ref{thm:pr-sub-systems} and Lemma \ref{lem:cp-preservation} of this paper shows that a quantum structure preserving subsystem truncation method can be developed for the entire class of completely passive systems which guarantees that the truncation is also completely passive. The idea in \cite{Pet11}, and later proved to be true for all completely passive linear quantum stochastic systems in \cite{Nurd10b}, is that if we allow a symplectic similarity transformations, the transfer function  of these systems can always be realized by a purely cascade connection of completely passive systems without the need for any direct interaction Hamiltonians $H^d_{jk}$ between any sub-systems $j$ and $k$. Then the model reduction strategy proposed in \cite{Pet12} is to truncate some tail components in this cascade. Using the results of \cite{Nurd10b} and Theorem \ref{thm:pr-sub-systems} of this paper, a similar truncation strategy to \cite{Pet12} can thus be applied to all completely passive systems provided that $G_{\pi}$ and the truncated subsystem $G_{\pi,1}$ are both asympotically stable. 

\section{Co-diagonalizability of the controllability and observability Gramians and model reduction by quasi-balanced truncation}
\label{sec:qb-reduction}
In this section we will consider the question of when it is possible to have a balanced or an ``almost'' balanced realization of a linear quantum stochastic linear system under the restriction of similarity transformation by a symplectic matrix. That is, we will derive conditions under which there is a symplectic similarity transformation of the system matrices $A,B,C,D$ such that the transformed system has controllability Gramian $P$ and observability Gramian $Q$ that are diagonal. Then we say that the Gramians $P$ and $Q$ are {\em co-diagonalisable}, the meaning of which will be made precise below. In the classical setting, if the system is minimal (i.e., it is controllable and observable) it is always possible to not only have the Gramians $P$ and $Q$ simultaneously diagonal but to make them diagonal and equal. The idea for model reduction by balanced truncation is to remove subsystems that are associated with the smallest positive diagonal entries of $P$ and $Q$, these correspond heuristically to systems modes that are least controllable as well as least observable. As will be shown, the restriction to a symplectic transformation somewhat limits what is achievable with linear quantum stochastic systems. Nonetheless, in Theorem \ref{thm:bt-q-lin} of this section precise conditions are deduced under which a symplectic transformation exists such that the transformed system will have $P$ and $Q$ simultaneously diagonal (though not necessarily to the same diagonal matrix).  

Consider a physically realizable $n$ degree of freedom linear quantum stochastic system (\ref{eq:qsde-out-quad}), thus the system matrices satisfy (\ref{eq:pr-1e})-(\ref{eq:pr-3e}), with $n_y$ possibly less than  $2m$ (i.e., possibly less outputs than inputs). We have seen that similarity transformations for linear quantum stochastic systems are restricted to symplectic matrices $T$ to preserve physical realizability. We assume that the system matrix $A$ is Hurwitz (all its eigenvalues are in the left half plane). As for classical linear systems, we can define the controllability and observability matrices as
$$
[\begin{array}{ccccc} B & AB & A^2 B & \ldots & A^{2n-1}B \end{array}],
$$
and 
$$
[\begin{array}{ccccc} C^{\top} & A^{\top}C^{\top} & (A^{\top})^2 C^{\top} & \ldots & (A^{\top})^{2n-1}C^{\top} \end{array}]^{\top},
$$
respectively. Since $A$ is Hurwitz, there exists a unique $0 \leq P = P^{\top} \in \mathbb{R}^{2n \times 2n}$ and $0 \leq Q=Q^{\top} \in \mathbb{R}^{2n \times 2n}$ satisfying the Lyapunov equations
\begin{eqnarray*}
&AP + P A^{\top} + B B^{\top} =0,\\
&A^{\top}Q +Q A + C^{\top} C =0,
\end{eqnarray*}
respectively, and, moreover, if the system is controllable (i.e., controllability matrix is full rank) and observable (i.e., observability matrix is full rank) then $P>0$ and $Q>0$; see, e.g., \cite{ZDG95}.  Using standard terminology, the matrices $P$ and $Q$ are referred to as the controllability and observability Gramian of the system, respectively. The transfer function $G(s)$ of the system  is defined as $G(s)=C(sI-A)^{-1}B+D$. In this section, we investigate a necessary and sufficient condition under which there is a {\em symplectic} matrix $T \in \mathbb{R}^{2n \times 2n}$ such that the transformed system with system matrices $(T A T^{-1}, TB, CT^{-1},D)$ has controllability and observability Gramians that are simultaneously diagonal. If there exists such a $T$ then we say that the Gramians $P$ and $Q$ are {\em co-diagonalizable}. A more convenient way to express co-diagonalizability is that there exists a symplectic matrix $T$ such that $TPT^{\top}=\Sigma_P$ and $T^{-\top}Q T^{-1}=\Sigma_Q$, with $\Sigma_P$ and $\Sigma_Q$  nonnegative and diagonal. In analogy with balanced realization for classical linear time-invariant systems \cite{ZDG95},  the case where $\Sigma_P=\Sigma_Q$ will be of particular interest. That is, when $P$ and $Q$ are co-diagonalizable to the same diagonal matrix.

Before stating the main results, let us introduce some formal definitions. Two matrices $M_1,M_2 \in \mathbb{R}^{2n \times 2n}$  are said to be {\em symplectically congruent} if there exists a  symplectic matrix $T \in \mathbb{R}^{2n \times 2n}$ such that $TM_1T^{\top}=M_2$. Two matrices $M_1,M_2 \in \mathbb{R}^{2n \times 2n}$ are said to be {\em symplectically  similar} if there exists a  symplectic matrix $T \in \mathbb{R}^{2n \times 2n}$ such that $TM_1T^{-1}=M_2$. Our first result is the following:

\begin{lemma}
\label{lem:sym-cong-sim} A real $2n \times 2n$ matrix $P=P^{\top} \geq 0$ is symplectically congruent to a real diagonal $2n \times 2n$ matrix $\Sigma \geq 0$ if and only if $\mathbb{J}_n P$ is symplectically similar to $\mathbb{J}_n \Sigma$. If the symplectic congruence holds and $\mathbb{J}_n P$ is diagonalizable then its eigenvalues come in imaginary conjugate pairs $\pm \imath \sigma_i$, $i=1,2,\ldots,n$. In particular, if $P>0$  then $P$ is symplectically congruent to  a diagonal matrix $\Sigma >0$, and $\mathbb{J}_n P$ is diagonalizable and symplectically similar to $\mathbb{J}_n \Sigma$.
\end{lemma}
\begin{proof}
See Appendix \ref{app:sym-cong-sim}.
\end{proof}

\textbf{Remark.} If $P \geq 0$ and $\mathbb{J}_n P$ is diagonalizable then the $n$ largest nonnegative eigenvalues $\sigma_1$, $\sigma_2$, $\ldots$, $\sigma_n$ of $\imath \mathbb{J}_n P$ are referred to as the {\em symplectic eigenvalues} of $P$. In particular, by Williamson's Theorem \cite{Will36}, \cite[Lemma 2]{PSL09}, $\mathbb{J}_nP$ is always diagonalizable when $P>0$ and in this case $\sigma_1$, $\sigma_2$, $\ldots$, $\sigma_n>0$. 

\begin{lemma}
\label{lem:sym-eig-struct} Let $P=P^{\top} \geq 0$ be a real $2n \times 2n$ matrix with $\mathbb{J}_nP$ diagonalizable (in particular, whenever $P>0$). Define $\mathbb{K}_n=P_s \mathbb{J}_n P_s^{\top}=\left[\begin{array}{cc} 0 & I_n \\ -I_n & 0 \end{array} \right]$ and $\tilde P=P_s P P_s^{\top}$, where $P_s$ is a $2n \times 2n$ permutation matrix acting as $$P_s(q_1,p_1,q_2,p_2,\ldots,q_n,p_n)^{\top} =(q_1,q_2,\ldots,q_n,p_1,p_2,\ldots,p_n)^{\top}.$$  Suppose that $P$ has symplectic eigenvalues $\sigma_1,\sigma_2,\ldots,\sigma_n$, with $\sigma_k \geq  0$ for $k=1,2,\ldots,n$. Then there exist  $2n$ linearly independent eigenvectors $v_1$, $v_1^{\#}$, $v_2$, $v_2^{\#}$, $\ldots$, $v_n$, $v_n^{\#}$ of $\mathbb{K}_n \tilde P$ satisfying $\mathbb{K}_n \tilde P v_k=\imath \sigma_k v_k$  and $\mathbb{K}_n \tilde P v_k^{\#}=-\imath \sigma_k v_k^{\#}$ for $k=1,2,\ldots,n$ such that the complex $2n \times 2n$ matrix 
\begin{equation}
V=[\begin{array}{cccccccc} v_1 & v_2 & \ldots & v_n & v_1^{\#} & v_2^{\#} & \ldots & v_n^{\#} \end{array}] \label{eq:V-transform}
\end{equation}
satisfies
\begin{eqnarray}
-\imath V^{-1}\mathbb{K}_n \tilde P V &=& {\rm diag}(\sigma_1,\sigma_2,\ldots,\sigma_n,-\sigma_1,-\sigma_2,\ldots,-\sigma_n), \label{eq:V-prop-1}\\
-\imath V^{\dag} \mathbb{K}_n V &=& {\rm diag}(I_n,-I_n). \label{eq:V-prop-2}
\end{eqnarray}
\end{lemma}
\begin{proof}
Note that $\mathbb{K}_n \tilde P = P_s   \mathbb{J}_n P_s^{\top} \tilde P = P_s \mathbb{J}_n P P_s^{\top}$. Therefore, $\mathbb{K}_n \tilde P$ and $\mathbb{J}_n P$ are similar to one another. Since $\mathbb{J}_n P$ is diagonalizable by hypothesis (in particular, whenever $P>0$), from Lemma \ref{lem:sym-cong-sim} it  follows that $-\imath \mathbb{K}_n \tilde P$ is diagonalizable with real eigenvalues $\pm \sigma_1$, $\pm \sigma_2$, $\ldots$, $\pm \sigma_n$ with the corresponding eigenvectors in $V$. Thus $-\imath \mathbb{K}_n \tilde P$ satisfies (\ref{eq:V-prop-1}). The lemma now follows immediately from the following result of \cite{Xiao09}:
\begin{lemma}
\cite[Lemma 71, Section VI, pp. 32-34]{Xiao09} If $\mathbb{K}_n \tilde P$ is diagonalizable then the matrix $V$ defined in (\ref{eq:V-transform}) satisfies (\ref{eq:V-prop-2}).
\end{lemma}
$\hfill \Box$
\end{proof}

Based on the above lemma we can prove the following:
\begin{theorem}
\label{thm:sym-diag} Let $P=P^{\top} \geq 0$ be a real $2n \times 2n$ matrix with $\mathbb{J}_n P$ diagonalizable (in particular, whenever $P>0$), and suppose that the symplectic eigenvalues of $P$ are $\sigma_1$, $\sigma_2$, $\ldots$, $\sigma_n$. Define $V$ and $P_s$ as in Lemma \ref{lem:sym-eig-struct}. Also, define the $2n \times 2n$ unitary matrix  
$$
U=\frac{1}{\sqrt{2}}{\rm diag}_n\left(\left[\begin{array}{cc} 1 & -\imath \\ 1 & \imath  \end{array} \right]\right),
$$ 
and the $2n \times 2n$ matrix $T=(P_s^{\top} V P_s U)^{\top}$. Then  $T$ is symplectic, $T^{-\top} \mathbb{J}_n P T^{\top}=\mathbb{J}_n \Sigma =\Sigma \mathbb{J}_n$, and $T P T^{\top}=\Sigma$, with $\Sigma={\rm diag}(\sigma_1 I_2, \sigma_2I_2,\ldots,\sigma_n I_2)$. 
\end{theorem}
\begin{proof}
See Appendix \ref{app:sym-diag}.
\end{proof}

\begin{theorem}
\label{thm:bt-q-lin} Let $G$ be a $n$ degree of freedom linear quantum stochastic system with system matrices $(A,B,C,D)$ with $A$ Hurwitz. Let $P=P^{\top}\geq 0$ and $Q=Q^{\top} \geq 0$ be, respectively, the controllability and observability Gramians of the system which are, respectively, the unique solution to the Lyapunov equations
\begin{eqnarray*}
AP+PA^{\top} + B B^{\top}=0,\\
A^{\top} Q + Q A + C^{\top} C=0.
\end{eqnarray*}
Suppose that $\mathbb{J}_nP$ and $\mathbb{J}_nQ$ are diagonalizable (in particular, whenever $P>0$ and $Q>0$) then the following holds: 

\begin{enumerate}
\item There exists  a symplectic matrix $T$ such that $TPT^{\top}=\Sigma$, $T^{-\top}QT^{-1}=\Sigma$, and $\Sigma= {\rm diag}(\sigma_1 I_2, \sigma_2 I_2,\ldots,\sigma_n I_2)$ for some $\sigma_1,\sigma_2,\ldots,\sigma_n \geq 0$, if and only if $\mathbb{J}_n P= Q \mathbb{J}_n$. In this case, $\sigma_1$, $\sigma_2$, $\ldots$, $\sigma_n$ are the coinciding symplectic eigenvalues of $P$ and $Q$ as well as the Hankel singular values of the system.

\item There exists a symplectic matrix $T$ such that $TPT^{\top}=\Sigma_P$, $T^{-\top}QT^{-1}=\Sigma_Q$, with $\Sigma_X$ ($X \in \{P,Q\}$) of the form $\Sigma_X= {\rm diag}(\sigma_{X,1} I_2, \sigma_{X,2} I_2,\ldots,\sigma_{X,n} I_2)$, with $\sigma_{X,1}$, $\sigma_{X,2}$, $\ldots$, $\sigma_{X,n}\geq 0$ the symplectic eigenvalues of $X$ (symplectic eigenvalues of $P$ need not be the same as those of $Q$), if and only if $[\mathbb{J}_n P,Q \mathbb{J}_n]=0$.  

\item There exists a symplectic matrix $T$ such that $TPT^{\top}=\Sigma_P$, $T^{-\top}QT^{-1}=\Sigma_Q$, for some real positive  semidefinite diagonal matrices $\Sigma_X$ ($X \in \{P,Q\}$),  if and only if there exist symplectic matrices $\tilde T_P$, $\tilde T_Q$, and diagonal symplectic matrices $D_P$ and $D_Q$ such that (i) $\tilde T_P P \tilde T_P^{\top}={\rm diag}(\sigma_{P,1}I_2,\sigma_{P,2}I_2,\ldots,\sigma_{P,n}I_2)$ with $\sigma_{P,1}$, $\sigma_{P,2}$, $\ldots$, $\sigma_{P,n}$ the symplectic eigenvalues of $P$, (ii) $\tilde T_Q^{-\top} Q \tilde T_Q^{-1}={\rm diag}(\sigma_{Q,1}I_2,\sigma_{Q,2}I_2,\ldots,\sigma_{Q,n}I_2)$ with $\sigma_{Q,1}$, $\sigma_{Q,2}$, $\ldots$, $\sigma_{Q,n}$ the symplectic eigenvalues of $Q$, and (iii) $D_P^{-1}\tilde T_P = D_Q \tilde T_Q$. 
\end{enumerate}

\end{theorem}
\begin{proof}
See Appendix \ref{app:bt-q-lin}.
\end{proof}

A discussion of the contents of the theorem is now in order. 

Point 1 of the theorem is the best  possible outcome and results in a direct analogue in the quantum case of  balanced realization. This is for two reasons. If $\mathbb{J}_n P = Q \mathbb{J}_n$ is satisfied then the Gramians $P$ and $Q$ can be co-diagonalized to the same diagonal matrix $\Sigma$. Moreover, the diagonal entries of $\Sigma$ come in identical pairs for each pair of conjugate position and momentum operators in the transformed system. This is desirable since when we discard oscillators from the model, we must simultaneusly remove pairs of conjugate position and momentum operators not just one of the pair. If the coefficients of $\Sigma$ were different for the position and momentum operators of the same oscillator, it is not possible to simply remove the operator corresponding to the larger value of the corresponding diagonal element of $\Sigma$. However, this ideal scenario is only achievable under the extremely restrictive condition that $\mathbb{J}_n P=Q\mathbb{J}_n$. Generic linear quantum stochastic systems will not satisfy this condition. Indeed, it is easy to generate random examples of linear quantum stochastic systems that will fail to meet this condition.

Point 2 of the theorem shows that it is possible to have co-diagonalization of $P$ and $Q$ to diagonal matrices $\Sigma_X$ of the form $\Sigma_X={\rm diag}(\sigma_{X,1} I_2,\sigma_{X,2} I_2,\ldots, \sigma_{X,n} I_2)$, $X\in\{P,Q\}$, but $\Sigma_P$ and $\Sigma_Q$ will not necessarily coincide. This weaker co-diagonalization is achievable under the weaker requirement (compared to the requirement of Point 1) that $[\mathbb{J}_n P,Q \mathbb{J}_n]=0$. Since $\Sigma_P$ and $\Sigma_Q$ need not coincide, their diagonal elements may not be ordered in the same way.  However, it will be shown in the next section that for quasi-balancesable system there is a natural strategy to truncate subsystems. Moreover, as will be demonstrated in a forthcoming example in Section \ref{sec:cp-systems-bt}, there exists a class of linear quantum stochastic systems that have a quasi-balanced realization.

Point 3 of the theorem is the weakest possible co-diagonalization result for $P$ and $Q$. This form of diagonalization can be achieved under a weaker condition than that of Points 1 and 2. It states that $P$ and $Q$ can be co-diagonalized by a symplectic matrix to, respectively, diagonal matrices $\Sigma_P$ and $\Sigma_Q$ which need not have the special form  stipulated in Points 1 and 2.  

\section{Truncation error bound in model reduction of quasi-balanceable systems}
\label{sec:truncation-error-bound}
In this section we shall derive a bound on the  magnitude of the error transfer function due to subsystem truncation of a quasi-balanceable linear quantum stochastic system. The error bound will be presented in Theorem \ref{thm:qb-error-bound} of this section. Let us introduce the notation  $\bar{\sigma}(\cdot)$ and $\lambda_{\rm max}(\cdot)$ to denote the largest singular value and eigenvalue of a matrix, respectively, with the matrix being square for the latter, and recall that the $H^{\infty}$ norm of a transfer function $H(s)$ is $\|H \|_{\infty}=\mathop{\sup}_{\omega \in \mathbb{R}} \bar{\sigma}(H(\imath \omega))$. We begin with the following lemma.

\begin{lemma}
\label{lm:zero-error} Let $G=(A,B,C,D)$ be a linear quantum stochastic system of degree $n$ with $A$ Hurwitz, $\mathbb{J}_n P$ and $\mathbb{J}_n Q$ diagonalizable, and $[\mathbb{J}_n P,Q\mathbb{J}_n]=0$. Let  $\Xi_G(s)=C(sI-A)^{-1}B+D$ be the transfer function of $G$, 
and let $T$ be a symplectic transformation such that $\tilde G =(TAT^{-1},TB,CT^{-1},D)$ is a quasi-balanced linear quantum stochastic system with diagonal positive semidefinite controllability and observability Gramians $\Sigma_P={
\rm diag}(\sigma_{P,1}I_{2},\sigma_{P,2}I_{2},\ldots,\sigma_{P,n}I_{2})$ and $\Sigma_Q={\rm diag}(\sigma_{Q,1}I_{2},\sigma_{Q,2}I_{2},\ldots,\sigma_{Q,n}I_{2})$, respectively. Partition the Gramian $\Sigma_X$ ($X \in \{Q,P\}$) as $\Sigma_X ={\rm diag}(\sigma_{X,r1},\sigma_{X,r2})$ with $\sigma_{X,r1} \in \mathbb{R}^{2r \times 2r}$ and $r<n$, and partition $\tilde A =TAT^{-1}$, $\tilde B=TB$, and $\tilde C = CT^{-1}$ compatibly as
$$
\tilde A=\left[\begin{array}{cc} \tilde A_{r,11} & \tilde A_{r,12} \\ \tilde A_{r,21} & \tilde A_{r,22} \end{array}\right];\, 
\tilde B=\left[\begin{array}{c} \tilde B_{r,1} \\ \tilde B_{r,2} \end{array}\right];\, \tilde C=\left[\begin{array}{cc} \tilde C_{r,1} & \tilde C_{r,2} \end{array}\right].
$$
Let $\Xi_{\tilde G_r}(s)= \tilde C_{r,1}(sI-\tilde A_{r,11})^{-1} \tilde B_{r,1}+D$. If $\tilde A_{r,11}$ is Hurwitz then for all $\omega \in \mathbb{R}$ 
\begin{eqnarray*}
\lefteqn{\bar{\sigma} (\Xi_G(\imath\omega)-\Xi_{\tilde G_{r}}(\imath\omega) )}\\
&=&\sqrt{\lambda_{\rm \max}((\Sigma_{P,r2}+\Delta_r(\imath \omega)^{-1}\Sigma_{P,r2}\Delta_r(\imath \omega)^{*})( \Delta_r(\imath \omega)^{-*}\Sigma_{Q,r2}\Delta_r(\imath\omega) + \Sigma_{Q,r2}))},
\end{eqnarray*}
with $\Delta_r(s)=sI-\tilde A_{r,22}-\tilde A_{r,21}(sI-\tilde A_{r,11})^{-1}\tilde A_{r,12}$. In particular, if either of, or both of, $P$ and $Q$ are singular with ${\rm rank}(PQ)=2\nu<2n$, and  $T$ has been chosen  such that $\Sigma_{P,\nu 1}\Sigma_{Q, \nu 1} >0$\footnote{If $T$ does not already satisfy this then pairs of consecutive odd and even indexed rows of $T$ can always be permuted to  get a new symplectic $T$ that does satisfy it to replace the original $T$.}, and    
$\tilde A_{r,11}$ is Hurwitz for $r=\nu,\nu+1,\ldots,n-1$, then $\|\Xi_{G}-\Xi_{\tilde G_{\nu}}\|_{\infty}=0$.
\end{lemma}
\begin{proof}
The expression for $\bar{\sigma}(\Xi_G(\imath \omega)-\Xi_{\tilde G_r}(\imath \omega))$ in the lemma follows mutatis mutandis from the derivation in Section 3 of \cite{Enns84}, with the obvious modifications. Now, by the hypothesis of the latter part of the lemma on $P$, $Q$, and $T$, we have that $\Sigma_{P,r2}\Sigma_{Q,r2}=0$ for all $r=\nu+1,\nu+2,\ldots,n-1$. Taking $r=n-1$, by the hypothesis that $\tilde A_{n-1,11}$ is Hurwitz we then get that $\bar{\sigma}(\Xi_G(\imath \omega)-\Xi_{\tilde G_{n-1}}(\imath \omega))=0$ for all $\omega$, therefore $\| \Xi_G -\Xi_{\tilde G_{n-1}}\|_{\infty}=0$. Since $\Xi_{\tilde G_{n-1}}$ has again, by construction, a quasi-balanced realization, the assumption that $\tilde A_{n-2,11}$ is Hurwitz implies analogously that $\|\Xi_{\tilde G_{n-1}} -\Xi_{\tilde G_{n-2}}\|_{\infty}=0$. Repeating this argument for $r=n-3,n-2,\ldots,\nu+1$, we obtain $\|\Xi_{\tilde G_{r}}-\Xi_{\tilde G_{r-1}}\|_{\infty}=0$ for $r=n-3,n-2,\ldots,\nu+1$. Therefore, with $\Xi_{\tilde G_n}=\Xi_G$,
$$
\| \Xi_G-\Xi_{\tilde G_{\nu}}\|_{\infty} =  \left\| \sum_{r=\nu+1}^{n}(\Xi_{\tilde G_r}  - \Xi_{\tilde G_{r-1}})\right\|_{\infty}\leq \sum_{r=\nu+1}^{n} \|\Xi_{\tilde G_r}  - \Xi_{\tilde G_{r-1}}\|_{\infty}=0.
$$
$\hfill \Box$
\end{proof}

The above lemma states that we can always discard subsystems corresponding to position and momentum pairs in the quasi-balanced realization that correspond to vanishing products $\sigma_{P,r} \sigma_{Q,r}$ without incurring any approximation error, provided the  submatrices  $\tilde A_{\nu,11},\tilde A_{\nu+1,11},\ldots, \tilde A_{n-1,11}$ are all Hurwitz, where $2\nu$ is the rank of $PQ$. Therefore, to simplify the exposition, from this point on we consider only the case where $G$ is   {\em minimal} in the usual sense that  $(A,B)$ is a controllable pair (i.e., the matrix $[\begin{array}{cccc} B & AB & \ldots & A^{2n-1}B\end{array}]$ is full rank) and $(A,C)$ is an observable  pair (i.e., the matrix $[\begin{array}{cccc} C^{\top} & A^{\top} C^{\top} & \ldots & (A^{2n-1})^{\top}C^{\top} \end{array}]$ is full rank). In this case we will have that $P>0$, $Q>0$ and  $PQ$ is nonsingular. We now show that when $[\mathbb{J}_n P, Q \mathbb{J}_n]=0$, a quasi-balanced realization of $G$ is similar to a non-physically realizable balanced realization of $G$ by a simple diagonal similarity transformation. This is stated precisely in the next lemma.
\begin{lemma}
\label{lm:qb-to-b} Let $\tilde G=(\tilde A,\tilde B,\tilde C,D)$ be a quasi-balanced realization of  $G=(A,B,C,D)$ as defined in Lemma \ref{lm:zero-error}, and suppose that $G$ is minimal. Let $T_b ={\rm diag}(T_{b,1}I_2,T_{b,2}I_2,\ldots,T_{b,n}I_2)$ be a diagonal matrix with $T_{b,j}=\left(\frac{\sigma_{Q,j}}{\sigma_{P,j}}\right)^{1/4}$. Then $\tilde G_b(s)=(\tilde A_b,\tilde B_b,\tilde C_b,D)$ with $\tilde A_b=T_b \tilde A T_b^{-1}$,  $\tilde B_b=T_b \tilde B$, and $\tilde C_b=\tilde C T_b^{-1}$, is a non-physically realizable balanced realization of $G$ (in particular, $\Xi_{\tilde G_b}(s)=\Xi_G(s)$) with diagonal and identical controllability and observability Gramians, $\Sigma_{P,b}=\Sigma_{Q,b}=\Sigma_b$, with $$\Sigma_b={\rm diag}(\sqrt{\sigma_{P,1}\sigma_{Q,1}}I_2,\sqrt{\sigma_{P,2}\sigma_{Q,2}}I_2,\ldots, \sqrt{\sigma_{P,n}\sigma_{Q,n}}I_2),$$  where $\Sigma_{P,b}$ and $\Sigma_{Q,b}$ denote the controllability and observability Gramians of $\tilde G_b$, respectively. Moreover, let $\tilde A$, $\tilde B$, $\tilde C$ be partitioned according to Lemma \ref{lm:zero-error}, and partition $\tilde A_b$, $\tilde B_b$, $\tilde C_b$ compatibly as
$$
\tilde A_b=\left[\begin{array}{cc} \tilde A_{b,r,11} & \tilde A_{b,r,12} \\ \tilde A_{b,r,21} & \tilde A_{b,r,22} \end{array}\right];\, 
\tilde B_b=\left[\begin{array}{c} \tilde B_{b,r,1} \\ \tilde B_{b,r,2} \end{array}\right];\, \tilde C_b=\left[\begin{array}{cc} \tilde C_{b,r,1} & \tilde C_{b,r,2} \end{array}\right],
$$
 and define $\Xi_{\tilde G_{b,r}}(s) = \tilde C_{b,r,1}(sI-\tilde A_{b,r,11})^{-1} \tilde B_{b,r,1}+D$, then
\begin{equation}
 \Xi_G(s) -\Xi_{\tilde G_r}(s) = \Xi_{G_b}(s) -\Xi_{\tilde G_{b,r}}(s), \label{eq:qb-n-b-equal-error}
\end{equation}
where $\tilde G_r$ is as defined in Lemma \ref{lm:zero-error}.
\end{lemma}
\begin{proof}
Note that from the given definitions of $T_b$ and $\Sigma_b$, $T_b$ is invertible and we easily verify that  $T_b \Sigma_P T_b^{\top}=\Sigma_b$ and  $T_b^{-\top} \Sigma_Q T_b^{-1}=\Sigma_b$. 
Since $TPT^{\top}=\Sigma_P$ and $T^{-\top} QT^{-1}=\Sigma_Q$,  defining $\tilde T_b=T_b T$ it follows that $\tilde T_b P \tilde T_b^{\top} = T_b \Sigma_P T_b^{\top}=\Sigma_b$ and $\tilde T_b^{-\top} Q\tilde T_b^{-1}=T_b^{-\top}\Sigma_Q T_b^{-1}=\Sigma_b$. 
Hence the system $\tilde G_b$ as defined in the lemma is similar to $G$ (via the transformation $\tilde T_b$) and has balanced Gramians $\Sigma_{P,b}=\Sigma_{Q,b}=\Sigma_b$,  therefore it is a balanced realization of $\Xi_G(s)$, although it is not physically realizable. Since $T_b$ is a diagonal matrix, we can partition it  conformably with the partitioning of $\tilde A_b$, $\tilde B_b$ and $\tilde C_b$ as given in the lemma as $T_b={\rm diag}(T_{b,r1},T_{b,r2})$ with $T_{b,r1}$ a diagonal and invertible $2r \times 2r$ matrix. By the diagonal form of $T_b$ we easily verify that $\tilde A_{b,r,11}= T_{b,r1} \tilde A_{r,11}T_{b,r1}^{-1}$, $\tilde B_{b,r,1}= T_{b,r1} \tilde B_{r,1}$, $\tilde C_{b,r,1}= \tilde C_{r,1}T_{b,r1}^{-1}$, and we  conclude that $\tilde G_{b,r}=(\tilde A_{b,r,11},\tilde B_{b,r,1}, \tilde C_{b,r,1},D)$ is similar to $\tilde G_r=(\tilde A_{r,11},\tilde B_{r,1}, \tilde C_{r,1},D)$ (via the transfomation $T_{b,r1}$) and thus $\Xi_{\tilde G_{b,r}}(s)=\Xi_{\tilde G_{r}}(s)$. From this and the fact established earlier that $\Xi_G(s)=\Xi_{\tilde G_b}(s)$, (\ref{eq:qb-n-b-equal-error}) therefore holds. $\hfill \Box$
\end{proof}
 
The identity (\ref{eq:qb-n-b-equal-error}) together with the fact that $\tilde G_b$ is a balanced realization of $\Xi_G(s)$ (although not physically realizable) allows us to immediately obtain bounds for the approximation error $\| \Xi_G -\Xi_{\tilde G_r}\|_{\infty}$ using standard proofs for results on error bounds for truncation of balanced realizations of  classical linear systems, see, e.g., \cite[Theorem 7.3]{ZDG95}. This is stated as the following theorem. 

\begin{theorem}
\label{thm:qb-error-bound} Let $G=(A,B,C,D)$ be a minimal linear quantum stochastic system of degree $n$ with $A$ Hurwitz, $\mathbb{J}_n P$ and $\mathbb{J}_n Q$ diagonalizable, and $[\mathbb{J}_n P,Q\mathbb{J}_n]=0$. Let  $\Xi_G(s)=C(sI-A)^{-1}B+D$ be the transfer function of $G$, and let $T$ be a symplectic transformation such that $\tilde G =(TAT^{-1},TB,CT^{-1},D)$ is a quasi-balanced linear quantum stochastic system 
 with diagonal positive definite controllability and observability Gramians $\Sigma_P$ and $\Sigma_Q$, respectively, 
and $\Sigma_b=(\Sigma_{P}\Sigma_{Q})^{1/2}={\rm diag}(\sigma_{b,1} I_{2i_1},\sigma_{b,2}I_{2i_2},\ldots,\sigma_{b,\mu}I_{2i_{\mu}})$ with $\sigma_{b,1}>\sigma_{b,2}>\ldots>\sigma_{b,\mu}>0$ for some positive integers $\mu\leq n$ and $i_1,i_2,\ldots,i_{\mu}$ such that $\sum_{r=1}^{\mu} i_r =n$. Let $\tilde A_{r,11}$, $\tilde G_{r}$, and $\Sigma_{X,r1}$ ($X \in \{Q,P\}$) be as defined in Lemma \ref{lm:zero-error}, and  let $j_r = \sum_{k=1}^r i_k$. Then for any $r < \mu$, $\tilde A_{j_r,11}$ is Hurwitz, and
$$
\| \Xi_G  - \Xi_{\tilde G_{j_r}}  \|_{\infty} \leq 2 \sum_{k=r+1}^{\mu} \sigma_{b,k},
$$
with the bound being achieved for $r=\mu-1$: $\| \Xi_G  - \Xi_{\tilde G_{j_{\mu-1}}} \|_{\infty}=\sigma_{b,\mu}$.
\end{theorem}

The error bound given by the above theorem gives a recipe for truncating the subsystems in a quasi-balanced realization of $G$. That is, one should truncate those subsystems in $\tilde G$ associated with position-momentum operator pairs that correspond to pairs $(\sigma_{P,r},\sigma_{Q,r})$ with the smallest geometric means $\sqrt{\sigma_{P,r}\sigma_{Q,r}}$. Furthermore, since $\tilde G_b$ is a balanced realization of $\Xi_G$, it turns out, rather nicely, that for quasi-balanced realizations of linear quantum stochastic systems these geometric means  in fact coincide with the Hankel singular values of $G$.

\section{Quasi-balanced truncation of completely passive linear quantum stochastic systems}
\label{sec:cp-systems-bt}
We now consider model reduction for  the special class of completely passive linear quantum stochastic systems as defined in Sec. \ref{sec:cp-systems}. The key result in this section is that members of this distinguished class have the property that, provided the $A$ matrix is Hurwitz, they always satisfy Point 2 of Theorem  \ref{thm:bt-q-lin} and thus always have a quasi-balanced realization. Therefore,   subsystem truncation can be performed on quasi-balanced realizations of this class of systems by removing subsystems associated with the smallest geometric means of the product of the diagonal controllability and observability Gramians,  with an error bound  given by Theorem \ref{thm:qb-error-bound}.

It has been shown in \cite{Nurd10b} that for completely passive systems the matrix $R$ has the block form $R=[R_{jk}]_{j,k=1,2,\ldots,n}$, where $R_{jk}$ is a $2 \times 2$ diagonal matrix of the form $R_{jk}=r_{jk} I_2$ for some $r_{jk} \in \mathbb{R}$. Also, if $\tilde K=[\tilde K_{ij}]_{i=1,2,\ldots,m,j=1,2,\ldots,n}$ with $\tilde K_{ij}=e^{\imath \theta_{ij}} \sqrt{\gamma_{ij}}$ with $\theta_{ij},\gamma_{ij} \in \mathbb{R}$ and $\gamma_{ij} >0$,  then by some straightforward algebra (see \cite[proof of Theorem 3.4]{JNP06}) we find that $B$ has the block form 
$B=[B_{ij}]_{i=1,2,\ldots,n,j=1,2,\ldots,m}$ with 
\begin{equation}
B_{ij}=-\sqrt{\gamma_{ji}}\left[\begin{array}{cc}  \cos(\theta_{ji}) & \sin(\theta_{ji}) \\ -\sin(\theta_{ji}) & \cos(\theta_{ji}) \end{array}\right]. \label{eq:B_ij}
\end{equation}
That is, $B_{ij}$ is a scaled rotation matrix on $\mathbb{R}^2$. These special structures of the matrices of completely passive linear quantum stochastic systems lead to the following results.

\begin{lemma}
\label{lem:cp-trans} If $G=(A,B,C,D)$ is a completely passive linear quantum stochastic system and $T$ is a unitary symplectic matrix, then the transformed system $\tilde G= (TAT^{-1},TB,CT^{-1},D)$ is also completely passive.
\end{lemma}
\begin{proof}
See Appendix \ref{app:cp-trans}.
\end{proof}

The above lemma essentially states that the complete passivity property is invariant under unitary symplectic similarity transformations. Also, we have that 
\begin{theorem}
\label{thm:cp-qbr} For any completely passive system that is asymptotically stable (i.e., the $A$ matrix is Hurwitz), $P=I$ and $[\mathbb{J}_n P,Q\mathbb{J}_n]=0$. That is, any such system has a quasi-balanced realization. In this case, the quasi-balancing transformation $T$ is unitary symplectic and can be determined by applying Theorem \ref{thm:sym-diag} to the observability Gramian $Q$ such that $T^{-\top}QT^{-1}=\Sigma_Q$. Moreover, any reduced system obtained by truncating a subsystem of the quasi-balanced realization is again completely passive.
\end{theorem}
\begin{proof}
See Appendix \ref{app:cp-qbr}.
\end{proof}

We are now ready to proceed to an example illustrating the use of Theorems \ref{thm:sym-diag}, \ref{thm:bt-q-lin}, and \ref{thm:cp-qbr}.

\begin{example}
Consider a two mirror optical cavity $G_1$ (the mirrors being labelled M1 and M2) with resonance frequency $\omega_c$ (say, in the order of GHz, its exact value not being critical here) and each mirror having decay rate $\gamma=12 \times 10^6$ Hz (typically a much smaller value than $\omega_c$ for a high Q cavity). The mirror  M2  is driven by coherent field $d\mathcal{A}_{\rm in}(t) = \alpha(t) e^{\imath \omega_c t} dt+d\mathcal{A}_{2}(t)$, where $\alpha(t)$ is a complex-valued signal and $\mathcal{A}_2(t)$ a vacuum annihilation field. For sufficiently large $t$, the light $\mathcal{A}_{\rm out}(t)$ reflected from M1 will be a filtered version (by the cavity) of $\mathcal{A}_{\rm in}(t)$ of the form $d\mathcal{A}_{\rm out}(t) = \tilde \alpha(t) e^{\imath\omega_c t} dt + d\mathcal{A}_1(t)$, where $\tilde \alpha(t)$ is a low-pass filtered version of $\alpha(t)$ (with some inherent vacuum fluctuations) and $\mathcal{A}_1(t)$  a vacuum annihilation field. 
Note that the light reflected from M2 (the other cavity output) is of no interest here since it contains a feedthrough of the unfiltered signal due to the  cavity being driven through this mirror, so we opt to ignore it. Working in a rotating frame with respect to the cavity resonance frequency $\omega_c$ (see, e.g., \cite{NJD08}), this two mirror cavity is described by a  one degree of freedom, 4 input, and 2 output linear quantum stochastic system with Hamiltonian matrix $R=0_{2 \times 2}$, coupling matrix $K=\frac{1}{2}\left[\begin{array}{cc} \sqrt{\gamma} & \imath\sqrt{\gamma} \\ \sqrt{\gamma} & \imath \sqrt{\gamma} \end{array} \right]$, and scattering matrix $S=I$, with the output from mirror M2 neglected.

It is possible to obtain a high roll-off rate and realize a sharper low-pass cut-off by connecting several  identical cavities together in a particular way, as shall now be described. Suppose that $G_2,G_3,\ldots,G_N$ are additional cavities all identical to $G_1$. For $j=1,2,\ldots,N-1$, connect the output from mirror M1 of cavity $G_{j-1}$ as input to mirror M2 of $G_{j}$. The signal to be filtered will drive mirror M2 of cavity $G_1$ and the output of interest will be the filtered light reflected off mirror M1 of cavity $G_N$. The optical low-pass filtering network $G_{{\rm net},N}$ composed of this interconnection of $G_1,G_2,\ldots,G_N$ is a linear quantum stochastic system with  $N$ degrees of freedom, $2(N+1)$ inputs (with a pair of quadratures being driven by the signal to be filtered), and 2 outputs\footnote{Physically there are actually $2(N+1)$ outputs but $2N$ of them are of no interest as they feed through the original unfiltered signal and are thus ignored.}. This network is  completely passive  since it is composed of completely passive cavities, and the $A$ matrix of the network is Hurwitz. For the case $N=5$, the Hankel singular values of the network\footnote{By Lemma \ref{lm:qb-to-b} and Theorem \ref{thm:cp-qbr},  they  coincide with the square root of the symplectic eigenvalues of $Q$ and come in identical pairs.} are  0.9028, 0.5826, 0.2632, 0.0812, and 0.0154 (each appearing twice). This suggests that  modes corresponding to the two smallest Hankel singular values may be removed without excessive truncation error. Transforming this system into quasi-balanced form by applying  Theorems \ref{thm:bt-q-lin} and \ref{thm:sym-diag}, and truncating the two modes corresponding to the two smallest Hankel singular values 0.0812 and 0.0154 gives a physically realizable asymptotically stable  reduced model $G_{\rm red,3}$ with three degrees of freedom, 12 inputs, and 2 outputs, with error bound $\|\Xi_{G_{\rm net,5}}-\Xi_{G_{\rm red,3}}\|_{\infty} \leq 2(0.0812+0.0154)=0.1932$. Here the driven input quadratures are labelled as the last two inputs $w_{2N+1}$ and $w_{2N+2}$, and the frequency responses of interest will be the ones from $w_{2N+1}$ and $w_{2N+2}$ to the  filtered output quadratures $y_1$ and $y_2$, respectively, with all other inputs only contributing vacuum fluctuations to the filtered signal\footnote{Note the (steady-state) vacuum fluctuations experienced by the $2N$ cavity quadratures will only be of unity variance, independently of $N$. This is because the steady-state (symmetrized) covariance of the fluctuations is given by the controllability Gramian $P$ and by  Theorem \ref{thm:cp-qbr} we have that for completely passive systems  $P=I$.}. Due to decoupling and symmetries in the cavity equations, the single input single output transfer functions  $w_{2N+1} \rightarrow y_1$ and $w_{2N+2} \rightarrow y_2$ are in fact identical, and their magnitude and phase frequency responses are shown in Fig.~\ref{fig:lp-modred}. It can be seen from the figure that the reduced model approximates the magnitude response quite well at lower frequencies but has a slower roll-off rate than the full network, as can be expected, and  also captures the phase response of the full model very well. 

\begin{figure}[tbph]
\label{fig:lp-modred}
\centering
\includegraphics[scale=0.6]{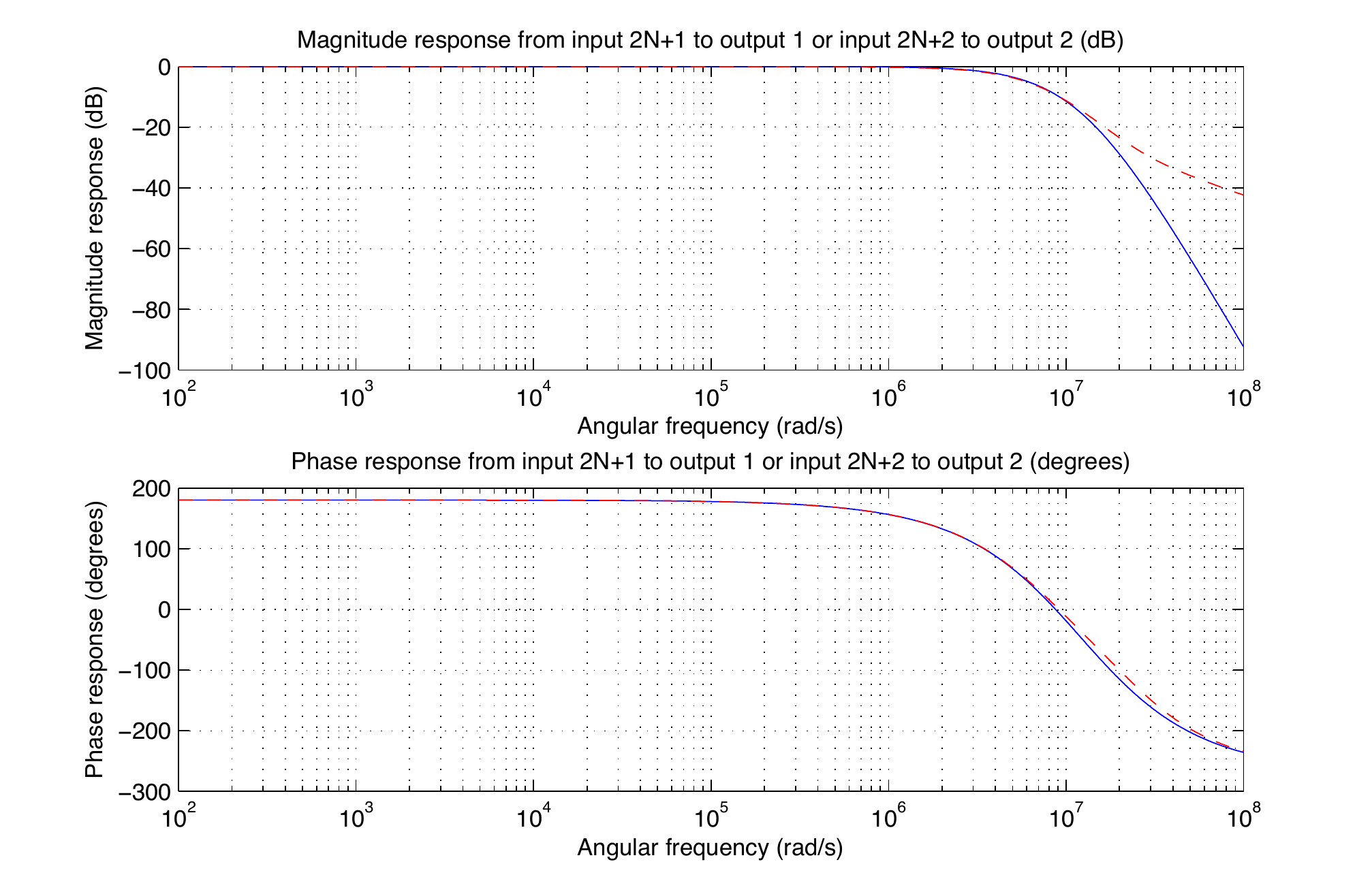}
\caption{Magnitude (top) and phase (bottom) frequency responses from the driven quadrature $w_{2N+1}$ to the filtered output $y_1$ (or, identically, from  $w_{2N+2}$ to the  filtered output $y_2$) of the full optical network and a three degree of freedom reduced model for $N=5$. The response of the full network and reduced model are indicated by solid blue lines and dashed red lines, respectively.}
\end{figure}

\end{example}

\section{Conclusion}
\label{sec:conclusion} This paper has developed several new results on model reduction of linear quantum stochastic systems. It is shown that the physical realizability and complete passivity properties of linear quantum stochastic systems are preserved under subsystem truncation. The paper also studied the co-diagonalizability of the controllability and observability Gramians of a linear quantum stochastic system. It is found that a balanced realization of the system, where the Gramians are diagonal and equal, exists if and only if a strong condition is satisfied, typically not satisfied by generic linear quantum stochastic systems. Necessary and sufficient conditions for weaker realizations with simultaneously diagonal controllability and observability Gramians were also obtained. The notion of a quasi-balanced realization of a linear quantum stochastic system was introduced and it is shown that the special class of asymptotically stable completely passive linear quantum stochastic systems always possess a quasi-balanced realization. An explicit bound for the truncation error of model reduction on a quasi-balanceable linear quantum stochastic system was also derived, in analogy with the classical setting. An example of an optical cavity network for optical low-pass filtering was developed  to illustrate the application of the results of this paper to model reduction of quasi-balanceable linear quantum stochastic systems.  

\section*{Appendices}
\appendices
\section{Proof of Lemma \ref{lem:sym-cong-sim}} 
\label{app:sym-cong-sim} Suppose that there is a symplectic matrix $T$ such that $TPT^{\top}=\Sigma$. Then we have (since $T^{\top}$ and $T^{-1}$ are also symplectic) that 
\begin{eqnarray*}
T^{-\top}\mathbb{J}_n P T^{\top} = (T^{-\top} \mathbb{J}_n T^{-1}) TP T^{\top} = \mathbb{J}_n TPT^{\top} = \mathbb{J}_n \Sigma,
\end{eqnarray*}
therefore $\mathbb{J}_n P$ is symplectically similar to $\mathbb{J}_n \Sigma$. 

Conversely, suppose that there is a symplectic matrix $T$ such that $T^{-\top} \mathbb{J}_nPT^{\top}=\mathbb{J}_n \Sigma$ for a real diagonal matrix $\Sigma$. Then we have
\begin{eqnarray*}
T P T^{\top} = -T\mathbb{J}_n \mathbb{J}_n P T^{\top} &=& -(T\mathbb{J}_n T^{\top}) T^{-\top} \mathbb{J}_nP T^{\top},\\
&=& -\mathbb{J}_n T^{-\top} \mathbb{J}_n P T^{\top},\\
&=& -\mathbb{J}_n\mathbb{J}_n \Sigma,\\
&=& \Sigma. 
\end{eqnarray*}
Therefore, $P$ is symplectically congruent to $\Sigma$. 

Suppose that $P \geq 0$ is symplectically congruent to $\Sigma$, and $\mathbb{J}_n P$ is diagonalizable. Then, by the above, $\mathbb{J}_n \Sigma$ is also diagonalizable. Furthermore, since $\Sigma \geq 0$, the matrix $\mathbb{J}_n \Sigma$ has eigenvalues of the form $\pm  \imath \sigma_1$, $\pm \imath \sigma_2$, $\ldots$, $\pm \imath \sigma_n$ (for some $\sigma_i \geq 0$ for $i=1,2,\ldots,n$). It follows that $\mathbb{J}_n P$ is also diagonalizable with the same set of eigenvalues. In the special case that $P>0$, then $\Sigma>0$, $\mathbb{J}_n P$ is diagonalizable by Williamson's Theorem \cite{Will36}, \cite[Lemma 2]{PSL09}, and $\sigma_i >0$  for $i=1,2,\ldots,n$. $\hfill \Box$

\section{Proof of Theorem \ref{thm:sym-diag}}
\label{app:sym-diag} Define $\tilde P$ as in Lemma \ref{lem:sym-eig-struct}. From the proof of Lemma \ref{lem:sym-eig-struct} we know that $\mathbb{K}_n \tilde P$ has eigenvalues $\pm \imath \sigma_1$, $\pm \imath \sigma_2$,$\ldots$,$\pm \imath \sigma_n$. Now, let
$$
W=P_s^{\top} V P_s = P_s^{\top} [\begin{array}{ccccccc} v_1 & v_1^{\#} & v_2 & v_2^{\#} &\ldots & v_n & v_n^{\#} \end{array}].
$$ 
Since $P=P^{\top} \geq 0$ and $\mathbb{J}_n P$ is assumed to be diagonalizable, we have from Lemma \ref{lem:sym-eig-struct} 
\begin{eqnarray*}
-\imath W^{-1} \mathbb{J}_n P W &=& -\imath P_s^{\top} V^{-1} P_s \mathbb{J}_n P P_s^{\top} V P_s,\\
&=& -\imath P_s^{\top} V ^{-1}(P_s \mathbb{J}_n P_s^{\top}) (P_s P P_s^{\top}) V P_s,\\
&=& P_s^{\top} (-\imath V^{-1} \mathbb{K}_n \tilde P V) P_s,\\
&=& P_s^{\top} {\rm diag}(\sigma_1,\sigma_2,\ldots,\sigma_n,-\sigma_1,-\sigma_2,\ldots,-\sigma_n) P_s,\\
&=& {\rm diag}(\sigma_1,-\sigma_1,\sigma_2,-\sigma_2,\ldots,\sigma_n, -\sigma_n),
\end{eqnarray*} 
Equivalently,
$
W^{-1}\mathbb{J}_n P W = {\rm diag}(\imath\sigma_1,-\imath\sigma_1,\imath \sigma_2,-\imath\sigma_2,\ldots,\imath \sigma_n, -\imath \sigma_n).
$
Moreover, we also have
\begin{eqnarray*}
-\imath W^{\dag} \mathbb{J}_n W &=& -\imath P_s^{\top} V^{\dag} (P_s \mathbb{J}_n P_s^{\top}) V P_s,\\
&=& P_s^{\top} (-\imath V^{\dag} \mathbb{K}_n V)P_s,\\
&=& P_s^{\top} {\rm diag}(I_n,-I_n) P_s,\\
&=& {\rm diag}_n \left({\rm diag}(1,-1) \right).
\end{eqnarray*}
Note that the unitary matrix $U$ in the statement of the theorem satisfies
$$
U^{\dag} {\rm diag}_n\left({\rm diag}(1,-1) \right) U = -\imath \mathbb{J}_n, 
$$
and also the matrix $ T_0=WU$ is real since
\begin{eqnarray*}
WU &=& \frac{1}{\sqrt{2}}P_s^{\top}[\begin{array}{ccc} v_1+ v_1^{\#} & -\imath v_1+ \imath v_1^{\#} & v_2+ v_2^{\#} \end{array} \\ 
&&\quad  \begin{array}{cccc} -\imath v_2+ \imath v_2^{\#} &\ldots & v_n+ v_n^{\#} & -\imath v_n+ \imath v_n^{\#} \end{array}].
\end{eqnarray*}
Thus we have that
\begin{eqnarray*}
 T_0^{\top} \mathbb{J}_n T_0 = T_0^{\dag} \mathbb{J}_n T_0 = U^{\dag} W^{\dag} \mathbb{J}_n W U = \imath U^{\dag} (-\imath  W^{\dag} \mathbb{J}_n W)U=\imath U^{\dag}{\rm diag}_n \left({\rm diag}(1,-1)\right) U =\mathbb{J}_n,
\end{eqnarray*}
and 
\begin{eqnarray*}
T_0^{-1} \mathbb{J}_n P T_0 &=& U^{\dag} W^{-1} \mathbb{J}_n P W U,\\
&=&\imath U^{\dag} (-\imath  W^{-1} \mathbb{J}_n PW)U,\\
&=& \imath U^{\dag} {\rm diag}(\sigma_1,- \sigma_1, \sigma_2,- \sigma_2,\ldots,\sigma_n, - \sigma_n) U,\\
&=& {\rm diag}_n(\sigma_1 \mathbb{J},\sigma_2\mathbb{J},\ldots,\sigma_n \mathbb{J}),\\
&=& \mathbb{J}_n \Sigma,\\
&=& \Sigma \mathbb{J}_n,
\end{eqnarray*}
with $\Sigma ={\rm diag}(\sigma_1 I_2,\sigma_2 I_2,\ldots,\sigma_n I_2)$. Thus we have constructed a symplectic matrix $T_0$ such that $T_0^{-1} \mathbb{J}_n P T_0=\mathbb{J}_n \Sigma = \Sigma \mathbb{J}_n$ (the second identity follows from the specific form of $\Sigma$). Defining $T=T_0^{\top}$ we have that $ T^{-\top} \mathbb{J}_n P T^{\top} =\mathbb{J}_n \Sigma = \Sigma \mathbb{J}_n$ and  from the proof of Lemma \ref{lem:sym-cong-sim} we  also conclude that $ T P T^{\top} = \Sigma$, as claimed. $\hfill \Box$

\section{Proof of Theorem \ref{thm:bt-q-lin}}
\label{app:bt-q-lin} We first prove the only if part of Point 1.  Suppose that there is a symplectic matrix $T$ such that 
$TPT^{\top}=\Sigma$, $T^{-\top}QT^{-1}=\Sigma$, and $\Sigma= {\rm diag}(\sigma_1 I_2, \sigma_2 I_2,\ldots,\sigma_n I_n)$ for some $\sigma_1,\sigma_2,\ldots,\sigma_n \geq 0$.  Then we have from Lemma \ref{lem:sym-cong-sim} that $T^{-\top}\mathbb{J}_n P T^{\top} = \mathbb{J}_n \Sigma$ and $T\mathbb{J}_n Q T^{-1}=\mathbb{J}_n\Sigma$. Now, note from Theorem \ref{thm:sym-diag} that $\mathbb{J}_n \Sigma= \Sigma \mathbb{J}_n$ (due to the specific form of $\Sigma$) from which it follows that $T^{-\top}Q \mathbb{J}_n T^{\top} = \mathbb{J}_n \Sigma$. Thus, we have that $T^{-\top}\mathbb{J}_n P T^{\top}=\mathbb{J}_n \Sigma = T^{-\top}Q \mathbb{J}_n T^{\top}$. It follows that $\mathbb{J}_n P=Q \mathbb{J}_n$.

For the if part of Point 1, suppose that $\mathbb{J}_n P =Q \mathbb{J}_n$. Let $\sigma_1$, $\sigma_2$, $\ldots$, $\sigma_n$ be the symplectic eigenvalues of $P$ and define $\Sigma=(\sigma_1 I_1,\sigma_2 I_2,\ldots,\sigma_n I_2)$. Then by Theorem \ref{thm:sym-diag} there exists a symplectic matrix $T$ such that $T^{-\top} \mathbb{J}_nP T^{\top}  =\mathbb{J}_n \Sigma$ and $T P T^{\top}=\Sigma$. Since $\mathbb{J}_n P =Q \mathbb{J}_n$, we also have that $T^{-\top} Q\mathbb{J}_n T^{\top}=\mathbb{J}_n \Sigma$ or, equivalently,  $T \mathbb{J}_n  Q T^{-1}=\mathbb{J}_n \Sigma$ (again using $\mathbb{J}_n \Sigma = \Sigma \mathbb{J}_n$). From this last equality it follows from Lemma \ref{lem:sym-cong-sim} that also $T^{-\top}Q T^{-1}=\Sigma$. 

Finally, we prove the last part of Point 1. It is apparent from the above that $P$ and $Q$ must have the same symplectic eigenvalues. Also note that $TPQT^{-1}= (TPT^{\top})(T^{-\top}QT^{-1})=\Sigma^2$. Since the eigenvalues of $PQ$ are squares of the Hankel singular values of $G$ and they are defined independently of the particular similarity transformation $T$ \cite{ZDG95}, $\sigma_1$, $\sigma_2$, $\ldots$, $\sigma_n$ are therefore also Hankel singular values of $G$. 

The proof of the only if part of Point 2  is similar to the proof of the only if part of Point 1, so we will leave the details for the reader. For the if part of Point 2,  note that since $\sigma_{X,1}$, $\sigma_{X,2}$, $\ldots$, $\sigma_{X,n}$ are the symplectic eigenvalues of $X$ for $X \in \{P,Q\}$,  $[\mathbb{J}_nP,Q\mathbb{J}_n]=0$ is, by Lemma \ref{lem:sym-cong-sim}, equivalent to $\mathbb{J}_n P$ and $Q\mathbb{J}_n$ being simultaneously diagonalizable by some complex matrix $W$ as
\begin{eqnarray*}
W^{-1} \mathbb{J}_n  P W &=& \imath {\rm diag}(\sigma_{P,1},-\sigma_{P,1},\ldots,\sigma_{P,n},-\sigma_{P,n}),\\
 W^{-1}  Q \mathbb{J}_n W &=& \imath {\rm diag}(\sigma_{Q,1},-\sigma_{Q,1},\ldots,\sigma_{Q,n},-\sigma_{Q,n}),
\end{eqnarray*}
In particular, the columns of $W$ are simultaneously eigenvectors of $\mathbb{J}_n P$ and $Q \mathbb{J}_n$. Following the corresponding steps in the proof of Theorem \ref{thm:sym-diag}, we can therefore establish that there is a symplectic matrix $T$ such that (again expoiting the specific form of $\Sigma_Q$ to commute it with $\mathbb{J}_n$)
\begin{eqnarray*}
T^{-\top} \mathbb{J}_n  P T ^{\top} &=& \mathbb{J}_n \Sigma_P ,\\
 T^{-\top} Q  \mathbb{J}_n T^{\top} &=& \mathbb{J}_n \Sigma_Q = \Sigma_Q\mathbb{J}_n \Leftrightarrow T \mathbb{J}_n Q T^{-1}=\mathbb{J}_n \Sigma_Q. 
\end{eqnarray*}
Therefore, from Lemma \ref{lem:sym-cong-sim} we conclude that $TPT^{\top}=\Sigma_P$ and $T^{-\top} Q T^{-1}=\Sigma_Q$.

Finally, we move on to proving Point 3. We first deal with the only if part. Suppose that there is a symplectic matrix $T$ such that 
$$TPT^{\top}=\Sigma_P={\rm diag}(\omega_{P,1},\omega_{P,2},\ldots,\omega_{P,2n-1},\omega_{P,2n}),$$ 
for some nonnegative numbers $\omega_{P,1},\omega_{P,2},\ldots,\omega_{P,2n-1},\omega_{P,2n}$, and 
$$T^{-\top} QT^{-1}=\Sigma_Q={\rm diag}(\omega_{Q,1},\omega_{Q,2},\ldots,\omega_{Q,2n-1},\omega_{Q,2n}),$$ 
for some nonnegative numbers $\omega_{Q,1},\omega_{Q,2},\ldots,\omega_{Q,2n-1},\omega_{Q,2n}$.  Since $\mathbb{J}_n X$ is assumed to be diagonalizable for $X \in \{P,Q\}$, by Lemma \ref{lem:sym-cong-sim} so is the matrix $\mathbb{J}_n \Sigma_X$. Moreover, since $\Sigma_X$ is real positive semidefinite, it follows (recall the proof of Lemma  \ref{lem:sym-cong-sim}) that  $\omega_{X,2i}=0$ if and only if $\omega_{X,2i-1}=0$ for $X\in \{P,Q\}$ and $i=1,2,\ldots,n$; for if this were not true then  $\mathbb{J}_n \Sigma_X$ will have zero as an eigenvalue with geometric multiplicity less than its algebraic multiplicity, contradicting the assumption that $\mathbb{J}_n X$ is diagonalizable. Now, for $X \in \{P,Q\}$, define 
$$
d_{X,2j-1}=\left\{ \begin{array}{cc} (\omega_{X,2j}/\omega_{X,2j-1})^{1/4} & \hbox{if $x_{2j-1} \neq0$ and $x_{2j}\neq 0$}\\
1 & \hbox{if $x_{2j-1} = 0$ and $x_{2j} = 0$}
\end{array} \right.,
$$ 
and 
$d_{X,2j}=\frac{1}{d_{X,2j-1}}$ for $j=1,2,\ldots,n$. Also, define $$D_{X}={\rm diag}(d_{X,1},d_{X,2},\ldots,d_{X,2n-1},d_{X,2n}),\; X \in \{P,Q\}.$$ Then notice that, by construction, $D_P$ and $D_Q$ are diagonal symplectic matrices. Moreover, $$D_P T P T^{\top} D_P  = {\rm diag}(e_{P,1}I_2,e_{P,2}I_2,\ldots,e_{P,n}I_2),$$ 
with $e_{P,i}=\sqrt{\omega_{P,2i-1}\omega_{P,2i}}$ for $i=1,2,\ldots,n$, and 
$$
D_Q T^{-\top} Q T^{-1} D_Q= {\rm diag}(e_{Q,1}I_2,e_{Q,2}I_2,\ldots,e_{Q,n}I_2),
$$
with $e_{Q,i}=\sqrt{\omega_{Q,2i-1}\omega_{Q,2i}}$ for $i=1,2,\ldots,n$. Define $\tilde T_P= D_P T$ and $\tilde T_Q= D_Q^{-1} T$ and note that by definition $D_P^{-1} \tilde T_P = D_Q \tilde T_Q$.  Again, it follows from Lemma \ref{lem:sym-cong-sim} that 
\begin{eqnarray*}
\tilde T_P^{-\top} \mathbb{J}_n P \tilde T_P^{\top} &=&\mathbb{J}_n {\rm diag}(e_{P,1}I_2,e_{P,2}I_2,\ldots,e_{P,n}I_n),\\
\tilde T_Q \mathbb{J}_n Q \tilde T_Q^{-1}&=&\mathbb{J}_n {\rm diag}(e_{Q,1}I_2,e_{Q,2}I_2,\ldots,e_{Q,n}I_2). 
\end{eqnarray*}
That is, $e_{X,1},e_{X,2},\ldots,e_{X,n}$ are the symplectic eigenvalues of $X$ for $X \in \{P,Q\}$. This completes the proof of the only if part.

Conversely, to prove the if part of Point 3, let $\sigma_{X,1},\sigma_{X,2},\ldots,\sigma_{X,n}$ be symplectic eigenvalues of $X \in \{P,Q\}$, and let $\tilde \Sigma_X={\rm diag}(\sigma_{X,1}I_2,\sigma_{X,2}I_2,\ldots,\sigma_{X,n}I_2)$. Suppose that there exist symplectic matrices $\tilde T_P$ and $\tilde T_Q$, and diagonal symplectic matrices $D_P$ and $D_Q$, such that $\tilde T_P P \tilde T_P^{\top}=\tilde \Sigma_P$ and $\tilde T_Q^{-\top} Q \tilde T_Q^{-1}=\tilde \Sigma_Q$, and $D_P^{-1} \tilde T_P = D_Q \tilde T_Q$. Let $\Sigma_P= D_P^{-1} \tilde \Sigma_P D_P^{-1}$ and $\Sigma_Q=D_Q^{-1} \tilde \Sigma_Q D_Q^{-1}$, and note that both are diagonal since  $D_Q$ and $D_P$ are diagonal. Define $T= D_P^{-1} \tilde T_P$, so then also $T=D_Q \tilde T_Q$. It follows that $T P T^{\top} = \Sigma_P$ and $T^{-\top} Q T^{-1}=\Sigma_Q$. $\hfill \Box$

\section{Proof of Lemma \ref{lem:cp-trans}}
\label{app:cp-trans} In this part, we show that a completely passive system after a unitary symplectic transformation remains completely passive. Let $T \in \mathbb{R}^{2n}$ be unitary symplectic and let $\tilde x=T x$, with $\tilde x=(\tilde q_1, \tilde p_1, \ldots, \tilde q_n, \tilde p_n)^{\top}$. Since $T$ is symplectic, the operators $\tilde q_1, \tilde p_1, \ldots, \tilde q_n, \tilde p_n$ satisfy the same canonical  commutation relations as $q_1, p_1, \ldots, q_n, p_n$. Define the annihilation operators $\tilde a_i=\frac{1}{2}(\tilde q_i + \imath \tilde p_i)$, $i=1,2,\ldots,n$ and let $\tilde a =(\tilde a_1,\tilde a_2,\ldots,\tilde a_n)^{\top}$. Also define $D(a) = [\begin{array}{cc} a^{\top} & a^{\dag} \end{array}]^{\top}$ and $D(\tilde a) = [\begin{array}{cc} \tilde a^{\top} & \tilde a^{\dag} \end{array}]^{\top}$. We can write 
$$
D(\tilde a) = [\begin{array}{cc} \Sigma^{\top} & \Sigma^{\dag}\end{array}]^{\top}\tilde x= [\begin{array}{cc} \Sigma^{\top} & \Sigma^{\dag}\end{array}]^{\top}Tx,
$$ 
with 
$$
\Sigma=\frac{1}{2}\left[\begin{array}{ccccccc} 1 & \imath & 0 & 0 & \ldots & 0 & 0 \\
0 & 0 & 1 & \imath & \ldots & 0 & 0 \\
\vdots & \vdots & \vdots & \vdots & \ddots & \vdots & \vdots \\
0 & 0 & 0 & 0 & \ldots & 1 & \imath \end{array} \right].
$$
Since  $\sqrt{2} [\begin{array}{cc} \Sigma^{\top} & \Sigma^{\dag}\end{array}]^{\top}$) is unitary (see, e.g., \cite{Nurd10b} we have that 
$$
[\begin{array}{cc} \Sigma^{\top} & \Sigma^{\dag}\end{array}]^{-\top}=2 [\begin{array}{cc} \Sigma^{\dag} & \Sigma^{\top}\end{array}],
$$
and therefore, since $\tilde x = [\begin{array}{cc} \Sigma^{\top} & \Sigma^{\dag}\end{array}]^{-\top} D(a) = 2 [\begin{array}{cc} \Sigma^{\dag} & \Sigma^{\top}\end{array}] D(a)$,
$$
D(\tilde a) = 2[\begin{array}{cc} \Sigma^{\top} & \Sigma^{\dag}\end{array}]^{\top}T[\begin{array}{cc} \Sigma^{\dag} & \Sigma^{\top}\end{array}]D(a),
$$
The matrix $W=2[\begin{array}{cc} \Sigma^{\top} & \Sigma^{\dag}\end{array}]^{\top}T[\begin{array}{cc} \Sigma^{\top} & \Sigma^{\dag}\end{array}]$ is necessarily Bogoliubov \cite{GJN10}, but it is also complex unitary since $T$ is real unitary and $\sqrt{2}[\begin{array}{cc} \Sigma^{\top} & \Sigma^{\dag}\end{array}]^{\top}$ and $\sqrt{2}[\begin{array}{cc} \Sigma^{\dag} & \Sigma^{\top}\end{array}]$ are both unitary.  In particular, $W$ has the doubled up form \cite{GJN10}
$$
W=\left[\begin{array}{cc} W_1 & W_2 \\ W_2^{\#} & W_1^{\#} \end{array} \right],
$$
for some matrices $W_1,W_2 \in \mathbb{C}^{n \times n}$. Since $W$ satisfies $WW^{\dag}=I=W^{\dag}W$ (unitarity) and $W^{\dag}{\rm diag}(I,-I)W={\rm diag}(I,-I)=W{\rm diag}(I,-I)W^{\dag}$ (the Bogoliubov property \cite{GJN10}), it follows by straightforward algebra that
$W_1^{\dag}W_1 +W_2^{\top}W_2^{\#} = I$,  $W_1^{\dag}W_2+W_2^{\top}W_1^{\#}=0$, $W_1^{\dag}W_1-W_2^{\top}W_2^{\#}=I$, and $W_1^{\dag}W_2-W_2^{\top}W_1^{\#} = 0$, implying that $W_2=0$ and $W_1$ is unitary. That is, $W= {\rm diag}(W_1,W_1^{\#})$. Therefore, it follows that $\tilde a = W_1 a \Leftrightarrow a = W_1^{\dag} \tilde a$. Since the system was originally completely passive with Hamiltonian $H= \frac{1}{2} a^{\dag}\tilde R a$ and coupling vector $L= \tilde K a$, the transformed system after the application of $T$ has Hamiltonian operator $\tilde H = \frac{1}{2}\tilde a^{\dag} (W_1 \tilde R W_1^{\dag}) \tilde a$ and $\tilde L = (\tilde K W_1^{\dag})\tilde a$. Since $D$ is unchanged when $T$ is applied, the form of $\tilde H$, $\tilde L$, and $D$ implies that the transformed system is again completely passive.  $\hfill \Box$

\section{Proof of Theorem \ref{thm:cp-qbr}}
\label{app:cp-qbr} The proof will be split into three main parts: Parts A, B, and C.

{\bf Part A}. Note that due to the diagonal form of $R_{jk}$ (see Section \ref{sec:cp-systems-bt}) we can straightforwardly verify that $\mathbb{J}_n R + (\mathbb{J}_n R)^{\top}=0$. Moreover, from the physical realizability criterion we also have that
$$
4\mathbb{J}_n \Im\{K^{\dag}K\}\mathbb{J}_n + B\mathbb{J}_m B^{\top}=0.
$$
If $B$ satisfies $B=\mathbb{J}_n B \mathbb{J}_m^{\top}=-\mathbb{J}_n B \mathbb{J}_m$ then we get that 
\begin{equation}
\left. \begin{array}{c} 4\mathbb{J}_n \Im\{K^{\dag}K\} + B B^{\top} =0 \\
 4 \Im\{K^{\dag}K\} \mathbb{J}_n + B B^{\top} =0 \end{array} \right\} \Leftrightarrow 2  \mathbb{J}_n \Im\{K^{\dag}K\}+ 2  \Im\{K^{\dag}K\} \mathbb{J}_n + BB^{\top}=0. \label{eq:K-B-id}
\end{equation}
 Using the fact that $A=2\mathbb{J}_n(R+\Im\{K^{\dag}K\})$ and $\mathbb{J}_n R + (\mathbb{J}_n R)^{\top}=0$, (\ref{eq:K-B-id}) implies that $A+A^{\top} + BB^{\top}=0$. That is, if $-B=\mathbb{J}_n B \mathbb{J}_m$ then the Lyapunov equation $AP +PA^{\top} + BB^{\top} =0$ has the unique solution $P=I$, uniqueness following from the assumption that $A$ is Hurwitz. Now, it is a straightforward  exercise to verify from the form of $B$ given in Sec. \ref{sec:cp-systems-bt} for a completely passive linear quantum stochastic system that indeed $B=-\mathbb{J}_n B \mathbb{J}_m$. We conclude that for a completely passive system with $A$ Hurwitz the controllability Gramian is $P=I$.

{\bf Part B.} For a completely passive system the matrix $K$ in the coupling vector $L=Kx$ has the special form
$$
K=\left[\begin{array}{ccccccccc} M_1 & \imath M_1 & M_2 & \imath M_2 & \ldots & M_{n-1} & \imath M_{n-1} & M_n & \imath M_n \end{array} \right],
$$
for some column vectors $M_i \in \mathbb{C}^{m}$, $i=1,2,\ldots,n$. From this structure, direct inspection shows that $\Im\{K^{\dag}K\}$ has the block form $\Im\{K^{\dag}K\}=[Z_{ij}]_{i=1,2,\ldots,n,j=1,2,\ldots,m}$, with $Z_{ij}$ real $2 \times 2$ matrices of the special form
$$
Z_{ij} = \left[\begin{array}{cc} z_{1,ij} & z_{2,ij} \\ -z_{2,ij} & z_{1,ij}\end{array}\right].
$$
From this block structure and the block structure of $\mathbb{J}_n$ we have that $\Im\{K^{\dag}K\} \mathbb{J}_n - \mathbb{J}_n \Im\{K^{\dag}K\}=0$. Using this identity and the  property of $\mathbb{J}_nR$ exploited in Part A, it follows that $\mathbb{J}_n A \mathbb{J}_n  + A=0 \Leftrightarrow \mathbb{J}_n A \mathbb{J}_n   = - A$. 

Let us proceed to consider the case where $D=[\begin{array}{cc} I_{n_y} & 0_{n_y \times (2m-n_y)}\end{array}]$ (if $n_y=2m$ then $D=I_{2m}$). Consider the Lyapunov equation $A^{\top}Q+QA +C^{\top}C=0$. Since $C^{\top}=-\mathbb{J}_n B \mathbb{J}_m D^{\top}$ (by  physically realizability of the system) we have that $C^{\top}C=\mathbb{J}_n B \mathbb{J}_m D^{\top} D  \mathbb{J}_m B^{\top} \mathbb{J}_n$. 
Also, due to the special form assumed for $D$ we have that $D^{\top}\mathbb{J}_{n_y/2} = \mathbb{J}_m D^{\top}$ and it follows that 
\begin{eqnarray*}
\mathbb{J}_n C^{\top} C \mathbb{J}_n &=&\mathbb{J}_n B D^{\top} D B^{\top} \mathbb{J}_n,\\
&=& \mathbb{J}_n B  D^{\top} \mathbb{J}_{n_y/2} \mathbb{J}_{n_y/2}^{\top} D B^{\top} \mathbb{J}_n,\\
&=& -\mathbb{J}_n B \mathbb{J}_m D^{\top}  D \mathbb{J}_m B^{\top} \mathbb{J}_n,\\
&=& -C^{\top}C.
\end{eqnarray*}
Now consider the Lyapunov equation $A^{\top} Q+ Q A + C^{\top} C=0$. By multiplying this equation on the left and the right by $\mathbb{J}_n$ this equation can be rewritten as the Lyapunov equation $(\mathbb{J}_nA\mathbb{J}_n)^{\top} \bar{Q} + \bar{Q} (\mathbb{J}_nA\mathbb{J}_n) + \mathbb{J}_n C^{\top}C \mathbb{J}_n=0$, with $\bar{Q}=-\mathbb{J}_n Q \mathbb{J}_n$. Using the facts established earlier that $\mathbb{J}_n A \mathbb{J}_n=-A$ and $\mathbb{J}_n C^{\top} C \mathbb{J}_n=-C^{\top}C$, we see that the Lyapunov equation may be rewritten as $A^{\top} \bar{Q} + \bar{Q} A + C^{\top} C=0$. That is, $Q$ and $\bar{Q}$ are solutions of the same Lyapunov equation. Since $A$ is Hurwitz, the solution to this equation is unique and therefore $Q= \bar{Q} \Leftrightarrow Q=-\mathbb{J}_n Q \mathbb{J}_n$. Since we have established that $P=I$, we thus conclude that $[\mathbb{J}_nP,Q\mathbb{J}_n]=\mathbb{J}_n Q\mathbb{J}_n +Q=0$ when $D=[\begin{array}{cc} I_{n_y} & 0 \end{array}]$. Moreover, note in passing that since $\mathbb{J}_n P$ is diagonalizable and $Q\mathbb{J}_n$ commutes with $\mathbb{J}_n P$, we have that $Q\mathbb{J}_n$ is also diagonalizable  and therefore so is $\mathbb{J}_n Q$.

{\bf Part C.} Now, consider the general case where there exists a matrix $E$ such that the square matrix $\tilde D=[\begin{array}{cc} D^{\top} & E^{\top} \end{array}]^{\top}$ is unitary and symplectic. We note that the unitarity of $\tilde D$ implies that $DE^{\top}=0$ and $\tilde D^{-1}=\tilde D^{\top}$. Also, the sympletic property of $\tilde D$ implies that $\tilde D^{-1}$ and $\tilde D^{\top}$ are symplectic. Define $\tilde B = B \tilde D^{-1}=B\tilde D^{\top}$.  Then we have that $\tilde B \tilde B^{\top} = BB^{\top}$ and $\tilde B \mathbb{J}_n \tilde B^{\top}= B \mathbb{J}_n B^{\top}$. It follows from this that $P=I$ is also the unique solution to the Lyapunov equation $A P+ P A + \tilde B \tilde B^{\top}=0$ and, since the system is physically realizable, $A \mathbb{J}_n + \mathbb{J}_n A^{\top} + \tilde B \mathbb{J}_m \tilde B^{\top}=0$. Let $D_0=[\begin{array}{cc} I_{n_y} & 0_{n_y \times (2m-n_y)} \end{array}]$. We now show that $\mathbb{J}_n  C^{\top} = \tilde B \mathbb{J}_m D_0^{\top}$. Indeed, we have
\begin{eqnarray*}
\mathbb{J}_n C^{\top} = B\mathbb{J}_m D^{\top}=\tilde B \tilde D \mathbb{J}_m D^{\top} = \tilde B (\tilde D \mathbb{J}_m  \tilde D^{\top}) \tilde D  D^{\top} = \tilde B \mathbb{J}_m D_0^{\top},
\end{eqnarray*}
where the last equality follows from the fact that $\tilde D  D^{\top}=D_0$ (by the unitarity of $\tilde D$). We  thus conclude that the system $\tilde G$ with system matrices $(A,\tilde B,C,D_0)$ is a physically realizable linear quantum stochastic system whose  controllability Gramian $P=I_{2n}$ and observability Gramian $Q$ coincides with the original system  $G$ with system matrices $(A,B,C,D)$.  Due to the special form of $D_0$, we conclude from Part B that $[\mathbb{J}_n P,Q\mathbb{J}_n]=0$. 

Finally, that the quasi-balancing transformation $T$ can be obtained by applying Theorem \ref{thm:sym-diag} to $Q$ such that $T^{-\top}Q T^{-1}=\Sigma_Q$ follows from the fact that $[\mathbb{J}_nP,Q\mathbb{J}_n]=0$ along the lines of the proof of Point 2 of Theorem \ref{thm:bt-q-lin}. Moreover, that $T$ is also unitary follows from the observation that $TT^{\top}=I$, since $P=I$ and $\Sigma = I$ (i.e., all the symplectic eigenvalues of $P$ are ones). Also, by Lemma \ref{lem:cp-trans}, the quasi-balanced realization obtained after applying $T$ is again complete passive. Therefore, from Lemma \ref{lem:cp-preservation} it now follows that the  reduced system obtained after applying subsystem truncation  is completely passive. $\hfill \Box$

\bibliographystyle{ieeetran}
\bibliography{ieeeabrv,rip,mjbib2004,irpnew}

\end{document}